\documentclass[a4paper,UKenglish,cleveref, autoref]{lipics-v2019}

\title{Grundy Coloring \& friends, Half-Graphs, Bicliques}
\titlerunning{}

\author{Pierre Aboulker}{DI/ENS, PSL University, Paris, France}{pierreaboulker@gmail.com}{}{}

\author{\'{E}douard Bonnet}{Univ Lyon, CNRS, ENS de Lyon, Université Claude Bernard Lyon 1, LIP UMR5668, France}{edouard.bonnet@ens-lyon.fr}{https://orcid.org/0000-0002-1653-5822}{}

\author{Eun Jung Kim}{Universit\'{e} Paris-Dauphine, PSL University, CNRS UMR7243, LAMSADE, Paris, France}{eun-jung.kim@dauphine.fr}{https://orcid.org/0000-0002-6824-0516}{}

\author{Florian Sikora}{Universit\'{e} Paris-Dauphine, PSL University, CNRS UMR7243, LAMSADE, Paris, France}{florian.sikora@dauphine.fr}{https://orcid.org/0000-0003-2670-6258}{}

\authorrunning{P. Aboulker, \'E. Bonnet, E.~J.~Kim, F. Sikora}

\Copyright{Pierre Aboulker, Édouard Bonnet, Eun Jung Kim, Florian Sikora}



\category{}

\supplement{}

\acknowledgements{A part of this work was done while the authors attended the ``2019 IBS Summer research program  
on Algorithms and Complexity in Discrete Structures'', hosted by the IBS discrete mathematics group.
The third and the fourth authors were partially supported by the ANR project “ESIGMA” (ANR-17-CE23-0010)}

\nolinenumbers 

\hideLIPIcs  

\EventEditors{John Q. Open and Joan R. Access}
\EventNoEds{2}
\EventLongTitle{42nd Conference on Very Important Topics (CVIT 2016)}
\EventShortTitle{CVIT 2016}
\EventAcronym{CVIT}
\EventYear{2016}
\EventDate{December 24--27, 2016}
\EventLocation{Little Whinging, United Kingdom}
\EventLogo{}
\SeriesVolume{42}
\ArticleNo{23}

\usepackage{xspace}
\usepackage[table]{xcolor}
\usepackage{tikz}
\usepackage{array}
\usepackage{mathdots}
\usepackage{pdfpages}

\usetikzlibrary{fit}
\usetikzlibrary{calc}
\usetikzlibrary{shapes}
\usetikzlibrary{decorations.pathreplacing}
\usepackage{algpseudocode,algorithm,algorithmicx}

\algrenewcommand\algorithmicrequire{\textbf{Precondition:}}
\algrenewcommand\algorithmicensure{\textbf{Postcondition:}}
\usepackage{caption}
\usepackage{subcaption}
\usepackage{verbatim}
\usepackage{todonotes}

\newcommand{\kmIS}{\textsc{$k$-Multicolored Independent Set}\xspace}
\newcommand{\kmSI}{\textsc{$k$-Multicolored Subgraph Isomorphism}\xspace}
\newcommand{\skmIS}{\textsc{$k$-MIS}\xspace}
\newcommand{\gr}{\textsc{Grundy Coloring}\xspace}
\newcommand{\pgr}{\textsc{Partial Grundy Coloring}\xspace}
\newcommand{\bchr}{\textsc{$b$-Chromatic Core}\xspace}

\newcommand{\cF}{\mathcal{F}}
\newcommand{\cC}{\mathcal{C}}

\theoremstyle{plain}

\newtheorem{observation}[theorem]{Observation}

\definecolor{g1}{rgb}{0,0,1}
\definecolor{g2}{rgb}{0.2,0,0.8}
\definecolor{g3}{rgb}{0.4,0,0.6}
\definecolor{g4}{rgb}{0.6,0,0.4}
\definecolor{g5}{rgb}{0.8,0,0.2}
\definecolor{g6}{rgb}{1,0,0}

\newcommand{\defparproblem}[4]{
 \vspace{1mm}
\noindent\fbox{
 \begin{minipage}{0.96\textwidth}
 \begin{tabular*}{\textwidth}{@{\extracolsep{\fill}}lr} #1 & {\bf{Parameter:}} #3 \\ \end{tabular*}
 {\bf{Input:}} #2 \\
 {\bf{Question:}} #4
 \end{minipage}
 }
 \vspace{1mm}
}

\begin{document}

\maketitle

\begin{abstract}
The \emph{first-fit} coloring is a heuristic that assigns to each vertex, arriving in a specified order $\sigma$, the smallest available color. 
The problem \gr\ asks how many colors are needed for the most adversarial vertex ordering $\sigma$, i.e., the maximum number of colors that the first-fit coloring requires over all possible vertex orderings.
Since its inception by Grundy in 1939, \gr\ has been examined for its structural and algorithmic aspects. 
A brute-force $f(k)n^{2^{k-1}}$-time algorithm for \gr\ on general graphs 
is not difficult to obtain, where $k$ is the number of colors required by the most adversarial vertex ordering. 
It was asked several times whether the dependency on $k$ in the exponent of $n$ can be avoided or reduced, and its answer seemed elusive until now. 
We prove that \gr is W[1]-hard and the brute-force algorithm is essentially optimal under the Exponential Time Hypothesis, thus settling this question by the negative. 

The key ingredient in our W[1]-hardness proof is to use so-called \emph{half-graphs} as a building block to transmit a color from one vertex to another. 
Leveraging the half-graphs, we also prove that \bchr\ is W[1]-hard, whose  parameterized complexity was posed as an open question by Panolan et al. [JCSS '17].
A natural follow-up question is, how the parameterized complexity changes in the absence of (large) half-graphs.
We establish fixed-parameter tractability on $K_{t,t}$-free graphs for \bchr\ and \pgr, making a step toward answering this question. 
The key combinatorial lemma underlying the tractability result might be of independent interest.
\end{abstract}

\section{Introduction}\label{sec:intro}

A coloring is said \emph{proper} if no two adjacent vertices receive the same color.
The \emph{chromatic number} of a graph $G$ denoted by $\chi(G)$ is the minimum number of colors required to properly color $G$.
Let us now consider a natural heuristic to build a proper coloring of a graph $G$. 
Given an ordering $\sigma$ of the vertices of $G$, consider each vertex of $G$ in the order $\sigma$ and assign to the current vertex the smallest possible color (without creating any conflict), i.e., the smallest color \emph{not} already given to one of its already colored neighbors. 
The obtained coloring is obviously proper and it is called a \emph{first-fit} or \emph{greedy} coloring.
The Grundy number, denoted by $\Gamma(G)$, is the \emph{largest} number of colors used by the first-fit coloring on some ordering of the vertices of $G$.
Thus $\Gamma(G)$ is an upper-bound to the output of a first-fit heuristic.

The Grundy number has been introduced in 1939 \cite{Grundy39}, but was formally defined only forty years ago, independently by Christen and Selkow~\cite{Christen} and by Simmons~\cite{Simmons}.
\gr in directed graphs already appears as an NP-complete problem in the monograph of Garey and Johnson \cite{gj79}.
The undirected version remains NP-hard on bipartite graphs~\cite{HSj} and their complements~\cite{Zaker-cobip}, chordal graphs~\cite{Sampaio12} and line graphs~\cite{HMY12}.
When the input is a tree, \gr can be solved in linear time~\cite{HHB82}. 
This result is generalized to bounded-treewidth graphs with an algorithm running in time $k^{O(w)}2^{O(wk)}n=O(n^{3w^2})$ for graphs of treewidth~$w$ and Grundy number $k$~\cite{TP97}, but this cannot be improved to $O^*(2^{o(w\log w)})$ under the ETH~\cite{BonnetFKS18}.
It is also possible to solve \gr in time $O^*(2.443^n)$~\cite{BonnetFKS18}.

In 2006, Zaker~\cite{Zaker06} observed that since a minimal witness (we will formally define a witness later) for Grundy number $k$ has size at most $2^{k-1}$, the brute-force approach gives an algorithm running in time $f(k)n^{2^{k-1}}$, that is, an XP algorithm in the words of parameterized complexity.
Since then it has been open whether \gr can be solved in FPT time, i.e., $f(k)n^{O(1)}$ (where the exponent does not depend on $k$).
FPT algorithms were obtained in chordal graphs, claw-free graphs, and graphs excluding a fixed minor~\cite{BonnetFKS18}, or with respect to the dual parameter $n-k$~\cite{HSj}.
The parameterized complexity of \gr in general graphs was raised as an open question in several papers in the past decade~\cite{Sampaio12,HS13,Gastineau14,BonnetFKS18}.

Closely related to Grundy coloring is the notion of \emph{partial Grundy coloring} and \emph{$b$-coloring}. 
We say that a proper coloring $V_1\uplus \cdots \uplus V_k$ is 
a \emph{partial Grundy coloring} of order $k$ if there exists $v_i\in V_i$ for each $i\in [k]$ such that 
$v_i$ has a neighbor in every $V_j$ with $j<i$. The problem \pgr takes a graph $G$ and a positive integer $k$, and 
asks if there is a partial Grundy coloring of order $k$.
Erd\H{o}s et al.~\cite{Erdos} showed that the partial Grundy number coincides with the so-called \emph{upper ochromatic} number.
This echoes another result of Erd\H{o}s et al.~\cite{Erdos1987} that Grundy number and ochromatic number (introduced by Simmons~\cite{Simmons}) are the same.

The \emph{$b$-chromatic core of order $k$} of a graph $G$ is a vertex-subset $C$ of $G$ with the following property: 
$C$ admits a partition into $V_1\uplus \cdots \uplus V_k$ such that there is $v_i\in V_i$ for each $i\in [k]$ 
which contains a neighbor in every $V_j$ with $j\neq i$. 
The goal of the problem \bchr\ is to determine whether an input graph $G$ contains a $b$-chromatic core of order $k$. 
This notion was studied in~\cite{EffantinGT16,PanolanPS17} in relation to \emph{$b$-coloring}, which is a proper coloring 
such that for every color $i$, there is a vertex of color $i$ which neighbors a vertex of every other color. 
The maximum number $k$ such that $G$ admits a $b$-coloring with $k$ colors is called the \emph{$b$-chromatic number} of $G$. 
In~\cite{PanolanPS17}, it was proven that deciding whether a graph $G$ has $b$-chromatic number at least $k$ is $W[1]$-hard parameterized by $k$.
The problem might be even harder since no polytime algorithm is known when $k$ is constant.
The authors left it as an open question whether \bchr is W[1]-hard or FPT.
 
\medskip

\noindent {\bf Our contribution: the half-graph is key.} 
We prove that \gr is W[1]-complete, thus settling the open question posed in~\cite{Sampaio12,HS13,Gastineau14,BonnetFKS18}. 
More quantitatively we show that the double-exponential XP algorithm is essentially optimal.
Indeed we prove that there is no computable function $f$ such that \gr is solvable in time $f(k)n^{o(2^{k-\log k})}$, unless the ETH fails.
This further answers by the negative an alternative question posted in \cite{Sampaio12,HS13}, whether there is an algorithm in time $n^{k^{O(1)}}$.

A key element in the hardness proof of \gr\ is what we call a \emph{half-graph} (definition in~\cref{subsec:halfgraph}). 
The main obstacle encountered when one sets out to prove W[1]-hardness of \gr\ is the difficulty of propagating a chosen color from a vertex to another while keeping the Grundy number low (i.e., bounded by a function of $k$).
Employing half-graphs turns out to be crucial to circumvent this obstacle, which we further examine in Section~\ref{sec:halfgrundy}.
Leveraging half-graphs as color propagation apparatus,  
we also prove that \bchr\ is W[1]-complete (albeit with a very different construction). This settles the question posed by~\cite{PanolanPS17}. 
 
\medskip

\noindent {\bf Our contribution: delineating the boundary of tractability.}
  All three problems, \gr, \pgr, and \bchr are FPT on nowhere dense graphs, for the parameter $k$ being the number of desired colors.
  The existence of each induced witness can be expressed as a first-order formula on at most $2^{k-1}$ variables in the case of \gr, and on at most $k^2$ variables in the case of \pgr and \bchr.
  The problem is therefore expressible in first-order logic as a disjunction of the existence of every induced witness while the number of induced witnesses is bounded by $2^{2^{2(k-1)}}$.
  And first-order formulas can be decided in FPT time on nowhere dense graphs~\cite{Grohe17}.
The next step is $K_{t,t}$-free graphs, i.e., those graphs without a biclique $K_{t,t}$ as a (not necessarily induced) subgraph, which is a dense graph class that contains nowhere dense graphs and graphs of bounded degeneracy.
In the realm of parameterized complexity, $K_{t,t}$-free graphs have been observed to admit FPT algorithms for otherwise W[1]-hard problems~\cite{TelleV12}.

We prove that \pgr and \bchr are fixed-parameter tractable on $K_{t,t}$-free graphs, even in the parameter $k+t$, now assuming that $t$ is not a fixed constant.
To this end, a combinatorial lemma plays a crucial role by letting us rule out the case when many vertices have large degree: if there are many vertices of large degree in a $K_{t,t}$-free graph, one can find a collection of $k$ vertex-disjoint and pairwise non-adjacent stars on $k$-vertices, which is a witness for \bchr and \pgr. 
Now, we can safely confine the input instances to have bounded degrees, save a few vertices. 
We present an FPT algorithm that works under this setting. 

\medskip

\noindent {\bf Limits and further questions.} Noting that the W[1]-hardness constructions for 
\gr and \bchr heavily rely on (large) half-graphs as a unit for color propagation, one may wonder 
how the parameterized complexity shall be affected if large half-graphs are excluded. 
The aforementioned tractability of \bchr is a step toward this end: recall that 
excluding a $K_{t,t}$ implies excluding a large half-graph $H_{t,t}$ (but not vice versa). 
We point out that our key combinatorial lemma for $K_{t,t}$-free graphs collapses on $H_{t,t}$-free graphs. Therefore, 
whether \gr and \bchr are fixed-parameter tractable on $H_{t,t}$-free graphs remains open. 
For \gr, whether it admits an FPT algorithm on $K_{t,t}$-free graphs is unsettled. 

The parameterized complexity of  \pgr on general graphs resisted our efforts so far.  
Again the main barrier toward a W[1]-hardness proof is to come up with a right unit for color propagation. 
Note that half-graphs cannot be used for \pgr as the partial Grundy number is already high on a large half-graph. 
This phenomenon suggests that the class of $H_{t,t}$-free graphs is a promising zone for inspection in order to establish the parameterized complexity of \pgr on general graphs.

\medskip

\noindent {\bf Organization of the rest of the paper.}
In~\cref{sec:prelim}, we lay out the terminology and preliminary results.
\cref{sec:barriers} exhibits the main issue in establishing the parameterized hardness of \gr.
In~\cref{sec:halfgrundy}, we show how to overcome this issue and prove that \gr is W[1]-hard.
In~\cref{sec:core}, being even more restricted in our gadget choice, we show that \bchr is W[1]-hard.
The tractability of \bchr and \pgr in $K_{t,t}$-free graphs is finally presented in~\cref{sec:tractability}.

\section{Preliminaries}\label{sec:prelim}

For any integer $i, j$, we denote by $[i,j]$ the set of integers that are at least $i$ and at most $j$, and $[i]$ is a short-hand for $[1,i]$.
We use the standard graph notations~\cite{DBLP:books/daglib/0030488}: for a graph $G$, $V(G)$ denotes the set of vertices of $G$, $E(G)$ denotes the set of edges. 
A vertex $u$ is a neighbor of $v$ if $uv$ is an edge of $G$.
The \emph{open neighborhood} of a vertex $v$ is the set of all neighbors of $v$ and $N[v]$ denotes the \emph{closed neighborhood} of $v$ defined as $N(v)\cup \{v\}$.
The open (closed, respectively) neighborhood of a vertex-set $S$ is $\bigcup_{v\in S} N(v) \setminus S$ ($\bigcup_{v\in S} N(v) \cup S$, respectively).
For a vertex-set $Y\subseteq V(G)$, we denote $N(v)\cap Y$ ($N[v]\cap Y$, respectively) simply as $N_Y(v)$ ($N_Y[v]$, respectively) and the same applies the open and closed neighborhood of a vertex-set $S$.
For two disjoint vertex-sets $X$ and $Y$, we say that $X$ is (anti-)complete with $Y$ if 
every vertex of $X$ is (non-)adjacent with every vertex of $Y$.

\subsection{Half-graphs}\label{subsec:halfgraph}

We call \emph{anti-matching} the complement of an induced matching.
The \emph{anti-matching of height $t$} is the complement of $t$ edges.
It will soon be apparent that all the coloring numbers considered in this paper are lowerbounded by $t$, in presence of an anti-matching of height $t$.
Therefore in the subsequent FPT reductions, we will not have the luxury to have anti-matchings of unbounded size.
This will constitute an issue since they are useful to propagate choices.
Imagine we have two sets $A$ and $B$ of size unbounded by the parameter, and we want to relate a choice in $A$ to the same choice in $B$.
Let us put an antimatching between $A$ and $B$.
Trivially independent sets of size 2 will correspond to consistent choices.
So it all boils down to expressing our problem in terms of finding large enough independent sets.
Now this option is not available, another way to propagate choices is to use \emph{half-graphs}.

We call \emph{half-graph} a graph whose vertices can be partitioned into $(A,B)$ such that there is no induced $2K_2$ in the graph induced by the edges with one endpoint in $A$ and the other endpoint in $B$, and $G[A]$ and $G[B]$ are both edgeless.
These graphs are sometimes called \emph{bipartite chain graphs}.
Equivalently we say that $(A,B)$ induces, or by a slight abuse of notation, \emph{is} a half-graph if $A$ and $B$ can be totally ordered, say $a_1, \ldots, a_{\lvert A \rvert}$ and $b_1, \ldots, b_{\lvert B \rvert}$ such that $N_B(a_1) \supseteq N_B(a_2) \supseteq \ldots \supseteq N_B(a_t)$ and if $a_ib_j$ is an edge then for every $j' \in [j+1,t]$, $a_ib_{j'}$ is also an edge.
The orderings $a_1, \ldots, a_{\lvert A \rvert}$ and $b_1, \ldots, b_{\lvert B \rvert}$ are called \emph{orders} of the half-graph.
Note that the orders of $A$ and $B$ ``orient the half-graph from $A$ to $B$''.
What we mean is that the same orderings do not testify that $(B,A)$ is a half-graph.
For that, we would need to take the reverse orders.
This technicality will have its importance in limiting the notion of ``path'' (or ``cycle'') of graphs, when we will impose that the orientation is the same throughout the path (or cycle).

\emph{The half-graph of height $t$} is a bipartite graph with partition $(A=\{a_1, \ldots, a_t\}, B=\{b_1, \ldots, b_t\})$ such that there is an edge between $a_i$ and $b_j$ if and only if $i < j$.
We denote this graph by $H_{t,t}$.
The \emph{level} of a vertex $v \in A$ ($v \in B$) in the half-graph of height $t$ is its index in the ordering of $A$ (or $B$).
Note that this is not uniquely defined for a half-graph in general, but it is for \emph{the} (canonical) half-graph of height $t$.
Any half-graph can be obtained from the half-graph of height $t$ (for some $t$) by duplicating some vertices.
The name \emph{half-graph} actually comes from Erd\H{o}s and Hajnal (see for instance \cite{Erdos84}).
More precisely what Erd\H{o}s defines as a half-graph corresponds in this paper to the (canonical) half-graph of height $t$.

We now define ``path'' and ``cycle'' of half-graphs.
A \emph{length-$\ell$ path of half-graphs} is a graph $H$ whose vertex-set can be partitioned into $(H_1,H_2,\ldots,H_{\ell+1})$ such that the three following conditions hold:
\begin{itemize}
  \item(i) there is no edge between $H_i$ and $H_j$ when $\lvert i-j \rvert \geqslant 2$,
  \item(ii) for every $i \in [\ell]$, $H[H_i \cup H_{i+1}]$ is a half-graph with bipartition $(H_i,H_{i+1})$, and 
  \item(iii) for every $i \in [2,\ell]$, the ordering of $H_i$ in the half-graph induced by $(H_{i-1},H_i)$ is the same as in the half-graph $(H_i,H_{i+1})$.
\end{itemize}
A \emph{length-$(\ell+1)$ cycle of half-graphs} is the same but there are edges between $H_1$ and $H_{\ell+1}$, namely a half-graph respecting the ordering of $H_1$ and $H_{\ell+1}$ in the half-graphs induced by $(H_1,H_2)$ and $(H_\ell,H_{\ell+1})$.
We denote by $H_{\ell \times t}$ the \emph{length-$\ell$ path of half-graphs} where each half-graph is isomorphic to $H_{t,t}$, and $H^*_{\ell \times t}$ the \emph{length-$\ell$ cycle of half-graphs} where each half-graph is isomorphic to $H_{t,t}$.
The graph $H^*_{\ell \times t}$ is also a convenient way of propagating a ``one-among-$t$'' consistent choice, since the only maximum independent sets of $H^*_{\ell \times t}$ take all $\ell+1$ vertices at the same level. 

\subsection{Grundy coloring}

A proper coloring of $G$ with $k$ color classes $V_1 \uplus \cdots \uplus V_k$ is a \emph{Grundy coloring} of order $k$ if for each $i\in [2,k]$, every vertex $v_i\in V_i$ has a neighbor in every $V_j$ with $j<i$.
The \emph{Grundy number} $\Gamma(G)$ is defined as the largest $k$ such that $G$ admits a Grundy coloring of order $k$. 
We say that an induced subgraph $H$ of $G$ is a \emph{witness achieving} (color) $k$ if $H$ has a Grundy coloring of order at least $k$; in this case, we simply say that $H$ is a  \emph{$k$-witness} (also called \emph{atom} by Zaker~\cite{Zaker06} or \emph{critical}~\cite{GyarfasKL97}).
We say that a $k$-witness is \emph{minimal} if there is no proper induced subgraph of it whose Grundy number is at least $k$. 
A graph $G$ has Grundy number at least $k$ if and only if it contains a minimal $k$-witness as an induced subgraph~\cite{Zaker06}. 
This gives us an equivalent formulation of \gr, which we frequently employ as a working definition.

\defparproblem{\gr}{An integer $k > 0$, a graph $G$.}{$k$}{Is there a vertex-subset $S \subseteq V(G)$ such that $G[S]$ is a Grundy minimal $k$-witness?}

Notice that in a minimal $k$-witness, exactly one vertex receives color $k$. 
It is not difficult to see that a minimal $k$-witness is connected. 
A \emph{colored} witness is a witness together with a coloring of its vertices.
The coloring can equivalently be given by the vertex ordering, by the coloring function, or by the partition into color classes.
When $k$ is not specified, a witness for $G$ is the witness which allows a Grundy coloring of order $\Gamma(G)$.

Let $V_1 \uplus \cdots \uplus V_k$ be a Grundy coloring of order $k$.
We say that a vertex $u$ colored $c'$  \emph{supports} $v$ colored $c$ if $u$ and $v$ are adjacent and $c' < c$. 
A vertex $v$ colored in $c$ is said to be \emph{supported} if the colors of the  vertices supporting $v$ span all colors from 1 to $c-1$. 


It was observed that the largest minimal $k$-witness uses $2^{k-1}$ vertices~\cite{Zaker06}.
These witnesses are implemented by a family of rooted trees called \emph{binomial trees} (see for instance \cite{BonnetFKS18}).
The set of binomial trees $(T_k)_ {k \geqslant 1}$ is defined recursively as follows:
\begin{itemize}
\item $T_1$ consists of a single vertex, declared as the root of $T_1$.
\item $T_k$ consists of two binomial trees $T_{k-1}$ such that the root of the first one is a child of the root of the other. 
The root of the latter is declared as the root of $T_k$.
\end{itemize}

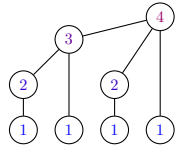
\begin{figure}[h!]
\centering
\begin{tikzpicture}[scale=0.6,transform shape]
\node[draw,circle] (r4) at (1,-0.5) {\textcolor{g4}{4}};
\node[draw,circle] (r3) at (-1,-1) {\textcolor{g3}{3}};
\node[draw,circle] (r32) at (-2,-2) {\textcolor{g2}{2}};
\node[draw,circle] (r321) at (-2,-3) {\textcolor{g1}{1}};
\node[draw,circle] (r31) at (-1,-3) {\textcolor{g1}{1}};
\node[draw,circle] (r2) at (0,-2) {\textcolor{g2}{2}};
\node[draw,circle] (r21) at (0,-3) {\textcolor{g1}{1}};
\node[draw,circle] (r1) at (1,-3) {\textcolor{g1}{1}};

\draw (r4) -- (r3);
\draw (r4) -- (r2);
\draw (r4) -- (r1);
\draw (r3) -- (r32);
\draw (r3) -- (r31);
\draw (r32) -- (r321);
\draw (r2) -- (r21);

\end{tikzpicture}
\caption{The binomial tree $T_4$, where the labels denote the color of each vertex in a first-fit coloring achieving the highest possible color.}
\label{fig:tk}
\end{figure}

We outline some basic properties of $k$-witnesses and binomial trees $T_k$. 
The first observation is straightforward.

\begin{observation}\label{obs:subset}
Any subset of $k'$ color classes of a $k$-witness, with $k' < k$, induces a $k'$-witness.
\end{observation}

The following is shown in a more general form in Lemma 7 of \cite{BonnetFKS18}.

\begin{lemma}\label{lem:subtree}
  Let $i \in [2,k-2]$, $X \subseteq V(T_k)$ be a subset of roots of $T_i$ whose parent is a root of $T_{i+1}$, and $T'_k$ be a tree obtained from $T_k$ by removing the subtree $T_{i-1}$ of every vertex in $X$.
  We assume that $T'_k$ is an induced subgraph of a graph $G$ and that $N(V(G) \setminus V(T'_k)) = X$.
  Then the three following conditions are equivalent in $G$:
  \begin{itemize}
  \item(i) There is a Grundy coloring that colors $k$ the root of $T'_k$.
  \item(ii) There is a Grundy coloring that colors $i$ every vertex of $X$ without coloring their parent in $T'_k$ first.
  \item(iii) There is a Grundy coloring that colors $i-1$ at least one neighbor of each vertex of $X$ without coloring any vertex of $T'_k$ first.
  \end{itemize}
\end{lemma}
\begin{proof}
  (iii) implies (ii), and (ii) implies (i) are a direct consequence of the optimum Grundy coloring of a binomial tree, as depicted in \cref{fig:tk}.
  We show that (i) implies (ii).
  This is equivalent to showing that the only way for a Grundy coloring of $T_k$ to color its root $k$, even when there is joker that enables us to give any color to a vertex of $X$, is to respect the coloring of \cref{fig:tk}.
  This holds since coloring a vertex of $X$ with a color greater than $i$ prevents from coloring its parent $w$ with color $i+1$.
  Indeed in that case $w$ cannot find a neighbor colored $i$ (which is not its own parent).
  Coloring a vertex of $X$ with a color smaller than $i$, simply will not work, since the Grundy coloring of $T_k$ that gives color $k$ to its root is unique.
  Finally (ii) implies (iii), since for every vertex of $X$, its only neighbor that can obtain color $i-1$ and is not its parent is outside $T'_k$.
  For a complete proof, see Lemma 7 of \cite{BonnetFKS18}. 
\end{proof}

\begin{lemma}\label{lem:twin}
  If $u$ and $v$ are false twins in $G$, i.e., $N_G(u)=N_G(v)$, then $\Gamma(G) = \Gamma(G - \{v\})$.
\end{lemma}
\begin{proof}
  Let $V_1\uplus\cdots \uplus V_k$ be an arbitrary Grundy coloring.
  We first observe that $u$ and $v$ must be in the same color class.
  Suppose the contrary, and assume that the color $c$ of $u$ is higher than the color $c'$ of $v$. 
  Since $u$ is supported by a vertex of color $c'$ and $N(u)=N(v)$, $v$ has a neighbor of the same color.
  This contradicts that the Grundy coloring is a proper coloring.
  Finally, observe that $u$ and $v$ support and are supported by the same set of vertices, which proves the lemma.
\end{proof}

\begin{lemma}\label{lem:crownSmallDeg}
  Let $H$ be an induced subgraph of $G$ such that all the vertices of $N(V(H))$ have degree at most $s$.
  Then no vertex of $V(H)$ can get a color higher than $\Gamma(H)+s$ in a Grundy coloring of $G$.
\end{lemma}
\begin{proof}
  Suppose that a vertex $v$ of $H$ gets a color at least $\Gamma(H)+s+1$ in a Grundy coloring of $G$.  
  Consider a minimal colored witness $H'$ of $G$ achieving $\Gamma(H)+s+1$ in which $v$ is the unique vertex getting the color $\Gamma(H)+s+1$.
  The vertices of $N(V(H))$, having degree at most $s$, can get color at most $s+1$.
  Since we are dealing with a minimal colored witness achieving a color strictly larger than $s+1$, they can get color at most $s$.
  Now we remove the first $s$ color classes from our colored witness.
  By \cref{obs:subset} we obtain a Grundy coloring achieving color $\Gamma(H)+1$, and by the previous remark, it no longer intersects $N(V(H))$.
  The connected component containing $v$ of that new witness achieves color $\Gamma(H)+1$ inside $H$, a contradiction.
\end{proof}

An immediate corollary of~\cref{lem:crownSmallDeg} is that a vertex with only few neighbors of high-degree cannot receive a high color.
\begin{corollary}\label{cor:fewHighDegNeighbors}
  In any greedy coloring, a vertex with at most $t$ neighbors that have degree at most $s$ cannot receive a color higher than $s+t+1$. 
\end{corollary}

\subsection{Partial Grundy and b-Chromatic Core}

A proper coloring of $G$ with $k$ color classes $V_1 \uplus \cdots \uplus V_k$ is a \emph{partial Grundy coloring} of order $k$ if, for each $i \in [2,k]$, there exists $v_i \in V_i$ such that $v_i$ has at least one neighbor in every $V_j$ with $j < i$. 
A proper coloring of $G$ with $k$ color classes $V_1 \uplus \cdots \uplus V_k$ is a \emph{$b$-coloring} of order $k$ if, for each $i \in [k]$, there exists $v_i \in V_i$ such that $v_i$ has at least one neighbor in every $V_j$ with $j \neq i$.
The \emph{$b$-chromatic number}\footnote{It is not difficult to see that there is a $\chi(G)$-coloring of $G$ which is a $b$-coloring as well, where $\chi(G)$ is the chromatic number of $G$.} of $G$ is the maximum $k$ such that $G$ allows a $b$-coloring of order $k$.
A vertex-set $S$ of $G$ is called a \emph{$b$-chromatic core of order $k$} if $G[S]$ admits a $b$-coloring of order $k$.
It is easy to see that admitting a partial Grundy coloring of order $k$ is monotone under taking an induced subgraph. 

\begin{observation}
A graph $G$ admits a partial Grundy coloring of order at least $k$ if and only if there exists a vertex-set $S \subseteq V(G)$ such that $G[S]$ admits a partial Grundy coloring of order $k$.
\end{observation}
\begin{proof}
  The forward direction is trivial.
  Suppose that $G[S]$ has a partial Grundy coloring $V_1 \uplus \cdots \uplus V_k$ of order $k$.
  For each uncolored vertex $v \in V \setminus S$, let $c \in [k]$ be the highest color such that for every color $i\in [c]$, $v$ has a neighbor colored in $i$. 
  We assign color $c+1$ to $v$.
  It is straightforward to see that the coloring of $S \cup \{v\}$ is a partial Grundy coloring of order at least $k$. 
\end{proof}

Following from this observation, we can formally define \pgr as: 
\defparproblem{\pgr}{An integer $k > 0$, a graph $G$.}{$k$}{Is there a vertex-subset $S\subseteq V(G)$ such that $G[S]$ admits a partial Grundy coloring of order $k$?}

On the other hand, $b$-coloring is not monotone under taking induced subgraphs.
That is, $G$ might contain a $b$-chromatic core of order $k$ but $G$ does not allow a $b$-coloring of order at least $k$.
This leads us to the following monotone problem, which is distinct from deciding whether the $b$-chromatic number of $G$ is at least $k$.

\defparproblem{\bchr}{An integer $k > 0$, a graph $G$.}{$k$}{Is there a vertex-subset $S\subseteq V(G)$ such that $G[S]$ admits a $b$-coloring of order $k$?}

For both \pgr and \bchr, the subgraph of $G$ induced by $S$ is referred to as a \emph{$k$-witness} if $S\subseteq V(G)$ is a solution to the instance $(G,k)$. 
A $k$-witness $H$ is called a \emph{minimal $k$-witness} if $H - v$ is not a $k$-witness for every $v \in V(H)$. 

Let $V_1 \uplus \cdots \uplus V_k$ be a proper coloring of $G$. 
In the context of partial Grundy coloring ($b$-coloring, respectively), 
we say that a vertex $v$ colored $c$ is \emph{supported by $u$} if $uv \in E(G)$ and $u$ is colored $c' < c$ ($c' \neq c$, respectively). 
In the partial Grundy coloring ($b$-coloring, respectively), a vertex $v$ colored $c$ is \emph{supported} if the colors of the supporting vertices of $v$ span all colors from 1 to $c-1$ (all colors of $[k] \setminus c$, respectively). 
Such a vertex $v$ is also called a \emph{center}.
A color $c$ is said \emph{realized} if a vertex $v$ colored $c$ is supported. 
That vertex $v$ is then \emph{realizing} color $c$.
Notice the crucial difference with Grundy colorings that these $c-1$ vertices do not need to be supported themselves.

Note that in a $k$-witness of \pgr or \bchr, each color class contains a supported vertex, which we call a \emph{center}. 
As each center requests at most $k-1$ supporting vertices, 
a minimal $k$-witnesses of \pgr or \bchr has size bounded by $k^2$~\cite{EffantinGT16}.
We note that a leaf of a center may very well be a center itself.
A caricatural example of that is a $k$-clique, where all the vertices are centers and leaves.
Therefore, both \pgr and \bchr can be seen as finding at most~$k^2$ vertices and a proper $k$-coloring of them to realize all $k$ colors.
We denote by $\Gamma'(G)$, respectively $\Gamma_b(G)$, the maximum integer $k$ such that $G$ admits a $k$-witness for \pgr, respectively \bchr. 

\subsection{ETH and handy W[1]-hard problems} \label{sec:W[1]}

The \emph{Exponential Time Hypothesis} (ETH) is a conjecture by Impagliazzo et al.~\cite{ImpagliazzoETH} that asserts that there is no $2^{o(n)}$-time algorithm for \textsc{3-SAT} on $n$-variable instances.
This conjecture implies, for instance, that FPT $\neq$ W[1].
Lokshtanov et al.~\cite{surveyETH} survey conditional lower bounds under the ETH.

We present some classic W[1]-hard problems which will be used as starting points of our reductions.
In the \kmIS problem, one is given a graph $G$, whose vertex-set is partitioned into $k$ sets $V_1, \ldots, V_k$, and is asked if there exists an independent set $I \subseteq V(G)$ such that $\lvert I \cap V_i \rvert =1$ for every $i \in [k]$.
The \kmIS problem is a classic W[1]-hard problem when parameterized by $k$ \cite{DF99}, and unless the ETH fails, cannot be solved in time $f(k){\lvert V(G)\rvert}^{o(k)}$ for any computable function $f$ \cite{Chen06}.
A related problem is \kmSI whose definition goes as follows:

\defparproblem{\kmSI}{An integer $k > 0$, a graph $G$ whose vertex-set is partitioned into $k$ sets $V_1, \ldots, V_k$, and a graph $H$ with $V(H)=[k]$.}{$k$}{Is there $\phi: i \in [k] \mapsto v_i \in V_i$ such that for all $ij \in E(H)$, $\phi(i)\phi(j) \in E(G)$?}

Even when $H$ is a 3-regular graph, \kmSI cannot be solved in time $f(k)|V(G)|^{o(k / \log k)}=f(|E(H)|)|V(G)|^{o(|E(H)| / \log |E(H)|)}$, unless the ETH fails (see Theorem 5.5 in \cite{MarxP15Arxiv}).
One can remove unnecessary edges in $G$ (in those $E(V_i,V_j)$ with $ij \notin E(H)$) such that $\phi([k])$ and $H$ are actually isomorphic in a solution.
\kmSI allows stronger ETH lower bounds than a reduction from \kmIS when the size of the new parameter depends on the number of edges of the sought pattern.

The \textsc{$k$-by-$k$ Grid Tiling}, or simply \textsc{Grid Tiling}, introduced by Marx \cite{Marx06,Marx07}, is a W[1]-hard problem (see for instance \cite{Cygan15}) which is a convenient starting point to show parameterized hardness for geometric problems \cite{Marx06,Feldmann18}.
It may turn out useful in other contexts, such as $H$-free graphs \cite{Bonnet18}, and as we will see in this paper for \bchr.

\defparproblem{\textsc{$k$-by-$k$ Grid Tiling} (or \textsc{Grid Tiling})}{Two integers $k, n > 0$, and a set $\{P_{i,j}\}_{i,j \in [k]}$ of $k^2$ subsets of pairs in $[n] \times [n]$.}{$k$}{Can one select exactly one pair $(x_{i,j},y_{i,j})$ in each $P_{i,j}$ such that for every $i, j \in [k]$, $x_{i,j} = x_{i,j+1}$ and $y_{i,j} = y_{i+1,j}$?}

Observe that the usual (equivalent) definition requires instead $x_{i,j} = x_{i+1,j}$ and $y_{i,j} = y_{i,j+1}$.
We prefer the other formulation since it is closer to the geometric interpretation of packing squares in the plane.

\section{Barriers to the Parameterized Hardness of \gr}\label{sec:barriers}

It is not difficult to see that deciding if a fixed vertex can get color $k$ in a greedy coloring is W[1]-hard.
Let us call this problem \textsc{Rooted Grundy Coloring}.

\begin{observation}\label{obs:rooted}
  \textsc{Rooted Grundy Coloring} is W[1]-hard.
\end{observation}
\begin{proof}  
  We design an FPT reduction from \kmIS to \textsc{Rooted Grundy Coloring}.
  Let $H$ be an instance of \skmIS with partition $V_1, \ldots, V_k$.
  We build an equivalent instance $G$ of \textsc{Rooted Grundy Coloring} in the following way.
  We copy $H$ in $G$ and we add a clique $C$ of size $k+1$.
  We call $v$ a fixed vertex of $C$ and we add a pendant neighbor $v'$ to $v$.
  We number the vertices of $C \setminus \{v\}$, $v_1, \ldots, v_k$, and we make $v_i$ adjacent to all the vertices of $V_i$ for each $i \in [k]$.
  A greedy coloring can color $v$ by $k+2$ if and only if there is a $k$-multicolored independent set in $H$.
\end{proof}

Of course this reduction does not imply anything for \gr.
Indeed the vertices of $V(H)$ could get much higher colors than $v$. 
This is precisely the issue with showing the parameterized hardness of \gr.

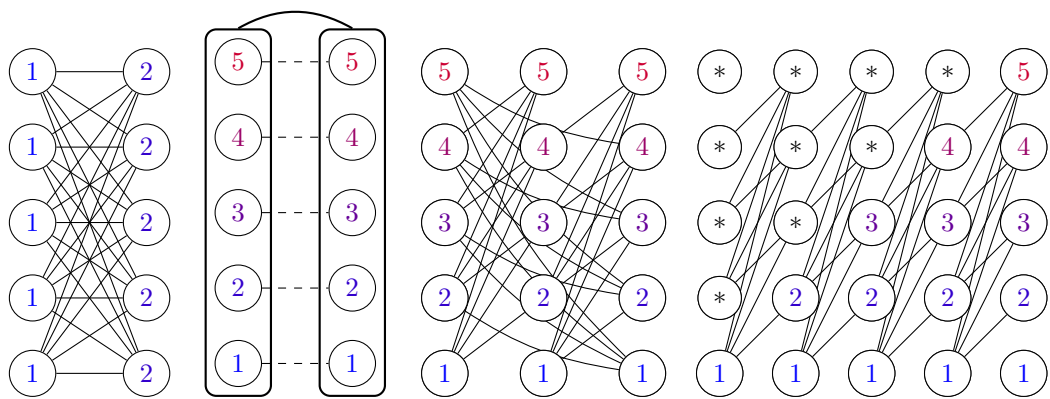
\begin{figure}[h!]
  \centering
  \begin{subfigure}[b]{0.18\textwidth}
    \centering
  \begin{tikzpicture}
    \def\t{5}
    \foreach \i in {1,...,\t}{
      \node[draw,circle] (a\i) at (0,\i) {\textcolor{g1}{$1$}} ;
      \node[draw,circle] (b\i) at (1.5,\i) {\textcolor{g2}{$2$}} ;
    }
    \foreach \i in {1,...,\t}{
      \foreach \j in {1,...,\t}{
        \draw (a\i) -- (b\j) ;
      }
    }
  \end{tikzpicture}
  \caption{Biclique}
  \label{subfig:bicliques}
\end{subfigure}
  \begin{subfigure}[b]{0.19\textwidth}
    \centering
  \begin{tikzpicture}
    \def\t{5}
    \foreach \i in {1,...,\t}{
      \node[draw,circle] (a\i) at (0,\i) {\textcolor{g\i}{$\i$}} ;
      \node[draw,circle] (b\i) at (1.5,\i) {\textcolor{g\i}{$\i$}} ;
      \draw[dashed] (a\i) -- (b\i) ;
    }
    \node[draw,rectangle,rounded corners,thick,fit=(a1)(a\t)] (A) {} ;
    \node[draw,rectangle,rounded corners,thick,fit=(b1)(b\t)] (B) {} ;
    \draw[thick] (A.north) to [bend left=30]  (B.north) ;
   \end{tikzpicture} 
  \caption{Anti-matching}
  \label{subfig:anti-matchings}
  \end{subfigure}
    \begin{subfigure}[b]{0.26\textwidth}
    \centering
    \begin{tikzpicture}
      \def\t{5}
      \def\l{3}
      \def\s{1.3}
      \pgfmathtruncatemacro\lm{\l-1}
      \pgfmathtruncatemacro\tm{\t-1}
      \foreach \j in {1,...,\l}{
        \foreach \i in {1,...,\t}{
          \node[draw,circle,fill=white] (a\j\i) at (\s * \j,\i) {\textcolor{g\i}{$\i$}} ;
        }
      }
      \foreach \j in {1,...,\lm}{
        \pgfmathtruncatemacro\jp{\j+1}
        \foreach \i in {1,...,\tm}{
          \pgfmathtruncatemacro\ip{\i+1}
          \foreach \h in {\ip,...,\t}{
            \draw (a\j\i) -- (a\jp\h) ;
          }
        }
      }
      \foreach \i in {1,...,\tm}{
        \pgfmathtruncatemacro\ip{\i+1}
          \foreach \h in {\ip,...,\t}{
            \draw (a\l\i) to [bend left=12] (a1\h) ;
          }
      }

      \foreach \j in {1,...,\l}{
        \foreach \i in {1,...,\t}{
          \node[draw,circle,fill=white] (a\j\i) at (\s * \j,\i) {\textcolor{g\i}{$\i$}} ;
        }
      }
   \end{tikzpicture}
 \caption{Half-graph short cycle}
 \label{subfig:cycle-HG}
\end{subfigure}
  \begin{subfigure}[b]{0.34\textwidth}
    \centering
    \begin{tikzpicture}
      \def\t{5}
      \def\l{5}
      \pgfmathtruncatemacro\lm{\l-1}
      \pgfmathtruncatemacro\tm{\t-1}
      \foreach \j in {1,...,\l}{
        \foreach \i in {1,...,\j}{
          \node[draw,circle] (a\j\i) at (\j,\i) {\textcolor{g\i}{$\i$}} ;
        }
      }
      \foreach \j in {1,...,\lm}{
        \pgfmathtruncatemacro\jp{\j+1}
        \foreach \i in {\jp,...,\t}{
          \node[draw,circle] (a\j\i) at (\j,\i) {$*$} ;
        }
      }
      \foreach \j in {1,...,\lm}{
        \pgfmathtruncatemacro\jp{\j+1}
        \foreach \i in {1,...,\tm}{
          \pgfmathtruncatemacro\ip{\i+1}
          \foreach \h in {\ip,...,\t}{
            \draw (a\j\i) -- (a\jp\h) ;
          }
        }
      }

        \foreach \j in {1,...,\l}{
        \foreach \i in {1,...,\j}{
          \node[draw,circle,fill=white] (a\j\i) at (\j,\i) {\textcolor{g\i}{$\i$}} ;
        }
      }
      \foreach \j in {1,...,\lm}{
        \pgfmathtruncatemacro\jp{\j+1}
        \foreach \i in {\jp,...,\t}{
          \node[draw,circle,fill=white] (a\j\i) at (\j,\i) {$*$} ;
        }
      }
   \end{tikzpicture} 
  \caption{Half-graph long path}
  \label{subfig:longChain-HG}
\end{subfigure}
\caption{The barriers in a propagation gadget for \gr. The biclique has small Grundy number but do not propagate, nor it imposes a unique choice. The three other propagations have arbitrary large Grundy number, rendering them useless for a parameterized reduction.}
\label{fig:barriers}
\end{figure}

A reduction starting from any W[1]-hard problem has to ``erase'' the potentially large Grundy number of the initial structure.
This can be done by isolating it with low-degree vertices.
However the degree $\Delta$ of the graph should be large, and a large chunk of the instance should have degree unbounded in $k$ since \gr is FPT parameterized by $\Delta+k$~\cite{Sampaio12,BonnetFKS18}.
Besides, as it is the case with W[1]-hardness reductions where induced subgraphs of the initial instance have to be tamed, we crucially need to propagate consistently one choice among a number of alternatives unbounded in the parameter.

A natural idea for encoding one choice among $t \ggg k$ is to have a set $S$ of $t$ vertices, one of which, the \emph{selected} vertex, receiving a specific color, say, 1.
Then a mechanism should ensure that one cannot color 1 two or more vertices of $S$.
Note that we cannot force that property by cliquifying $S$, as this would elevate the Grundy number to at least $t$.
Furthermore, by Ramsey's theorem, there will be independent sets of size $2^{\Omega(\frac{\log{t}}{k})}$ in $S$.
Thus we might as well assume that $S$ is an independent set, and look for another way of preventing two vertices from getting color 1, than by adding edges inside $S$.

We are now facing the following task: Given a bipartite graph, or a ``path'' or ``cycle'' of bipartite graphs whose partite sets are copies of $S$, ensure that exactly one vertex can receive color 1 in each partite set, and that this corresponds to a single vertex in $S$.  
A biclique certainly has low Grundy number (see \cref{subfig:bicliques}) but does not propagate nor it actually forces a unique choice.
Anything more elaborate seems to have large Grundy number, be it the complement of an induced matching, or \emph{anti-matching}, (see \cref{subfig:anti-matchings}), a ``cycle'' of half-graphs (see \cref{subfig:cycle-HG}), or even a long ``path'' of half-graphs (see \cref{subfig:longChain-HG}).
We remind the reader that, as detailed in \cref{sec:prelim}, half-graphs and anti-matchings are (the) two ways of propagating a consistent independent set.

\section{Overcoming the Barriers: Short Path of Half-Graphs}\label{sec:halfgrundy}

It might be guessed from the previous section that the solution will come from a constant-length ``path'' of half-graphs.
It is easy to see that half-graphs (that can be seen as length-one path of half-graphs) have Grundy number at most 3.
Due to the $2K_2$-freeness of the half-graph, there cannot be both color 1 and color 2 vertices present on both sides of the bipartition, say $(A,B)$.
If $A$ is the side missing a 1 or a 2 among its colors, then $B$ in turn cannot have a 3 (nor a 4).
The absence of vertices colored 3 in $B$ prevents vertices colored 4 in $A$.
Overall, no vertex with color 4 can exist. 

It takes more time to realize that a length-two path of half-graphs have constant Grundy number.
We keep this proof to convey a certain intuition behind the boundedness of the Grundy number of such graphs.
However we will then generalize this statement to any constant-length path.
\begin{lemma}\label{lem:chain-of-length-two}
The Grundy number of a length-two path of half-graphs is at most 5.
\end{lemma}
\begin{proof}
  We denote by $A$, $B$, and $C$ the tripartition of the vertex-set of $H$, a path of half-graphs of length two.
  A first observation is that for a color $i \geqslant 2$ to appear somewhere in $A \cup C$, there should be all the colors from 1 to $i-1$ already present in $B$.
  In particular for the Grundy number to exceed 5, one needs to have all the colors up to 4 present in $B$.
  Without loss of generality, we assume that one vertex $v \in B$ colored 2 is supported by a vertex $u \in A$ colored 1.

  We claim that there cannot be a vertex colored $2$ in $A$.
  Indeed this vertex would need to be supported by a vertex colored $1$ in $B$, thereby creating with $uv$ an induced $2K_2$ in the half-graph $H[A \cup B]$.
  This absence implies that every vertex $x \in B$ colored $3$ has its supporting vertex colored $2$, say $y$, in $C$.
  In turn it implies that there is no vertex colored 3 in $C$, since this vertex and its supporting $2$ in $B$ would form with $xy$ an induced $2K_2$ in the half-graph $H[B \cup C]$.
  Now a vertex colored $4$ in $B$ has to be supported by a $3$ in $A$, and by a $2$ in $C$.
  Similarly, these two facts prevent the existence of $4$ in $A$, and in $C$.
  In that case, there cannot be a $5$ nor a $6$ in $B$.
  The absence of a $5$ in $B$ makes it impossible to have a $6$ in $A \cup C$.
  Overall there cannot be a $6$ appearing in $H$.
\end{proof}

We could stop here and make a W[1]-hardness using path of half-graphs of length one or two.
For the sake of curiosity, but also in order to present a more transparent construction using length-four path of half-graphs, we show that, in all generality, the Grundy number of a length-$\ell$ path of half-graphs is a ``constant'' depending only on $\ell$.

\begin{lemma}\label{lem:chain-of-length-l}
  The Grundy number of a length-$\ell$ path of half-graphs is at most $4^{\ell}$.
\end{lemma}
\begin{proof}
  Achieving a (more) reasonable upper bound --the Grundy number of such graphs is most likely polynomial or even linear in $\ell$-- proves to be not so easy.
  We choose here to give a short proof of an admittingly bad upper bound.

  We show this bound by induction on $\ell$.
  Note that the statement trivially holds for $\ell=0$, and that we previously verified it for $\ell=1$.
  Assume that the Grundy number of any length-$(\ell-1)$ path of half-graphs is at most $4^{\ell-1}$, for any $\ell \geqslant 2$. 
  
  Let $G$ be a length-$\ell$ path of half-graphs, with partition $V(G) = V_0 \uplus V_1 \uplus \dots \uplus V_{\ell}$ where $G[V_i \cup V_{i+1}]$ is a half-graph for each $i \in [\ell-1]$.
  Observe that $G - V_0$ and $G - V_{\ell}$ are both length-$(\ell-1)$ path of half-graphs.
  Let $H$ be a colored witness of $G$ achieving color $\Gamma(G)$.
  We distinguish some cases based on the number of colors of $H$ appearing in $V_0$ or in $V_\ell$.
  In each case, we conclude with \cref{obs:subset}.
  No more than $4^{\ell-1}$ colors of $H$ can be missing in $V_0$ (resp.~in $V_{\ell}$).
  Otherwise by \cref{obs:subset}, the corresponding color classes form a $k$-witness $G-V_0$ (resp.~in $G-V_{\ell}$) with some $k > 4^{\ell-1}$, contradicting the induction hypothesis.

  So we may assume that at least $\Gamma(G)-4^{\ell-1}$ colors appear in $V_0$ (resp.~in $V_\ell$).
  Thus at least $(2\Gamma(G)-2 \cdot 4^{\ell-1})-\Gamma(G)=\Gamma(G) - 2 \cdot 4^{\ell-1}$ colors appears in both $V_0$ and $V_{\ell}$.
  If $\Gamma(G) > 4^\ell$, then $\Gamma(G) - 2 \cdot 4^{\ell-1} > 4^{\ell-1}$.
  We further claim that the corresponding color classes would form a witness in $G-V_0$, a contradiction. 
  If not, it must be because a vertex $x \in V_1$ colored $i$ was adjacent to a vertex $y \in V_0$ colored $j < i$, and is not adjacent to any vertex colored $j$ in $G-V_0$.
  But we know that $V_0$ contains a vertex $y'$ colored $i$, which in turn must be adjacent to a vertex $x' \in V_1$ colored $j$, forming an induced $2K_2$ in $G[V_0 \cup V_1]$, a contradiction.
  Therefore, $\Gamma(G) \leqslant 4^\ell$.

We observe that our proof works for a more general notion of ``path of half-graphs'' where one does not impose the orders of the successive half-graphs to have the same orientation (see the second paragraph of \cref{subsec:halfgraph}).
\end{proof}

Combining the ideas in \cref{lem:chain-of-length-two,lem:chain-of-length-l}, we obtain the following corollary, with a better bound than the straightforward application of \cref{lem:chain-of-length-l}.
\begin{corollary}\label{lem:chain-of-length-four}
  The Grundy number of a length-four path of half-graphs is at most 53.
\end{corollary}
\begin{proof}
  53 being $3 \times (3 \times 5 + 2) + 2$.
\end{proof}

We are now ready to present the hardness construction.
We reduce from \kmSI whose definition can be found in the preliminaries.

\begin{theorem}\label{hardness:main}
  \gr is W[1]-complete and, unless the ETH fails, cannot be solved in time $f(q)n^{o(2^{q-\log q})}$ (nor in time $f(q)n^{2^{o(q)}}$) for any computable function $f$, on $n$-vertex graphs with Grundy number~$q$.
\end{theorem}

\begin{proof}
  The membership to W[1] is given by the framework of Cesati~\cite{Cesati03}, since there is always a witness of size $2^{q-1}$.
  We show the W[1]-hardness of \gr by reducing from \kmSI with 3-regular pattern graphs.
  Let $(G=(V_1,\ldots,V_k,E),H=([k],F))$ be an instance of that problem. 
  We further assume that $k$ is a positive even integer and there is no edge between $V_i$ and $V_j$ in $G$ whenever $ij \notin E(H)$.
  The goal is now to find $v_1 \in V_1, \ldots, v_k \in V_k$ such that $H$ is isomorphic to  $G[\{v_1, \ldots, v_k\}]$.
  Even with these restrictions \kmSI cannot be solved in time $f(k)|V(G)|^{o(k / \log k)}=f(|E(H)|)|V(G)|^{o(|E(H)| / \log |E(H)|)}$, unless the ETH fails (see \cite{Marx10,MarxP15}, and Theorem 5.5 in \cite{MarxP15Arxiv}).
  
  We build an equivalent \gr-instance $(G',q)$ with $q = \lceil \log k \rceil + 55$ as follows.
  For each color class $V_i$, we fix an arbitrary total ordering $\leqslant_i$ on the vertices of $V_i$, and we write $u <_i u'$ if $u \neq u'$ and $u \leqslant_i u'$. 
  Let $i \in [k]$ and let $i(1), i(2), i(3) \in [k]$ be the three neighbors of $i$ in $H$.
  Each $V_i$ is encoded by a length-$4$ path of half-graphs denoted by $H_i$ (see~\cref{fig:encoding-Vi}). 
  We now detail the construction of $H_i$. 
  
  We set $V(H_i) := L_i \cup V_{i,i(1)} \cup V_{i,i(2)} \cup V_{i,i(3)} \cup R_i$. 
  The vertices of $L_i$ (resp.~$R_i$) are in one-to-one correspondence with the vertices of $V_i$.
  We denote by $l(u)$ (resp.~$r(u)$) the vertex of $L_i$ (resp.~$R_i$) corresponding to $u \in V_i$.
  For each $p \in [3]$, the vertices of  $V_{i,i(p)}$ are in one-to-one correspondence with the edges of $E(V_i,V_{i(p)})$.
  We denote by $z(u,v)$ the vertex of $V_{i,i(p)}$ corresponding to the edge $uv \in E(V_i,V_{i(p)})$ with $u \in V_i$ and $v \in V_{i(p)}$.
  
We set $E(H_i) := E(L_i, V_{i,i(1)}) \cup E(V_{i,i(1)}, V_{i,i(2)}) \cup E(V_{i,i(2)}, V_{i,i(3)}) \cup E(V_{i,i(3)}, R_i)$:
\begin{itemize}
\item $l(u)z(u',v) \in E(L_i, V_{i,i(1)})$ \hfill{if and only if $u <_iu'$}
\item for $p \in [2]$, $z(u,v)z(u',v') \in E(V_{i,i(p)}, V_{i,i(p+1)})$ \hfill{if and only if $u <_i u'$}
\item $z(u,v)r(u') \in E(V_{i,i(3)}, R_i)$ \hfill{if and only if $u <_i u'$.}
\end{itemize}
For each pair of vertices $u,u' \in V_i$ such that $u <_i u'$, we add an edge between $l(u)$ and $z(u',v) \in V_{i,i(1)}$, respectively $z(u,v) \in V_{i,i(1)}$ and $z(u',v') \in V_{i,i(2)}$, respectively $z(u,v) \in V_{i,i(2)}$ and $z(u',v') \in V_{i,i(3)}$, respectively $z(u,v) \in V_{i,i(3)}$ and $r(u')$ (see \cref{fig:encoding-Vi}). 

  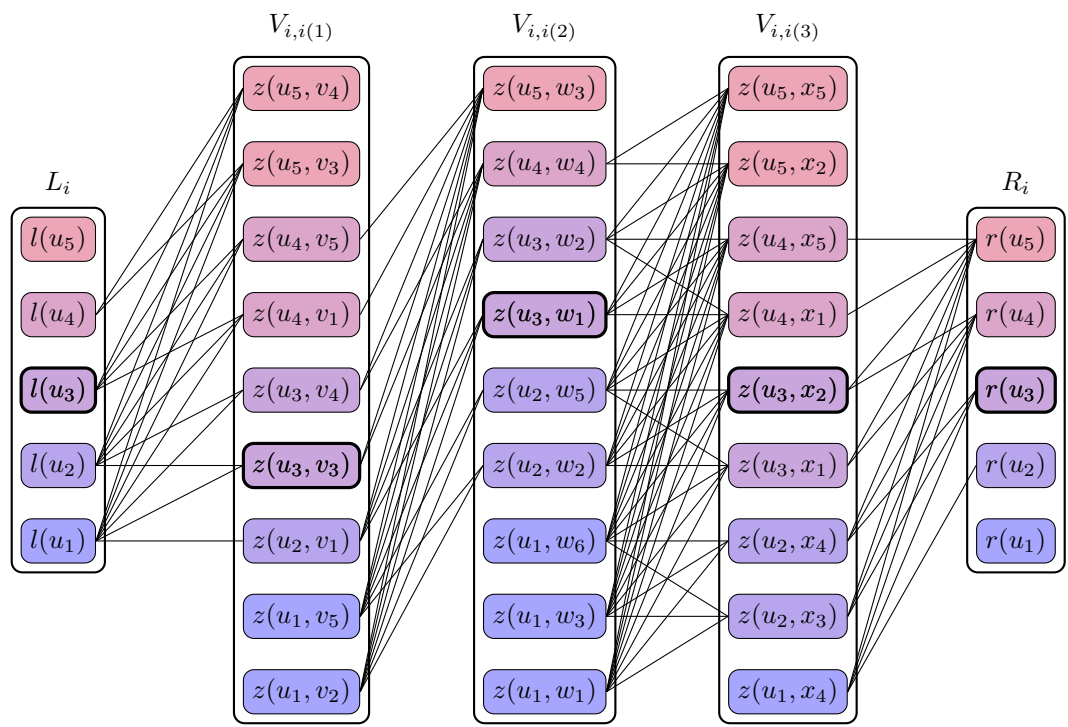
\begin{figure}
    \centering
    \begin{tikzpicture}
      \def\t{5}
      \foreach \i in {1,...,\t}{
        \node[rounded corners,fill opacity=0.35,fill=g\i,text opacity=1,draw] (l\i) at (0,\i) {$l(u_\i)$} ;
        \node[rounded corners,fill opacity=0.35,fill=g\i,text opacity=1,draw] (r\i) at (12.6,\i) {$r(u_\i)$} ;
      }
      \node[draw, thick, rectangle, rounded corners, fit=(l1) (l\t)] (l) {} ;
      \node[draw, thick, rectangle, rounded corners, fit=(r1) (r\t)] (r) {} ;
      \node at (0,5.75) {$L_i$} ;
      \node at (12.6,5.75) {$R_i$} ;
      \def\a{3.2}
      \foreach \i/\j/\x/\y in {
        1/-1/1/v_2, 1/0/1/v_5, 1/1/2/v_1, 1/2/3/v_3, 1/3/3/v_4, 1/4/4/v_1, 1/5/4/v_5, 1/6/5/v_3, 1/7/5/v_4,
        2/-1/1/w_1, 2/0/1/w_3, 2/1/1/w_6, 2/2/2/w_2, 2/3/2/w_5, 2/4/3/w_1, 2/5/3/w_2, 2/6/4/w_4, 2/7/5/w_3,
        3/-1/1/x_4, 3/0/2/x_3, 3/1/2/x_4, 3/2/3/x_1, 3/3/3/x_2, 3/4/4/x_1, 3/5/4/x_5, 3/6/5/x_2, 3/7/5/x_5
      }{
        \pgfmathtruncatemacro\jo{\j + 2}
        \node[rounded corners,fill opacity=0.35,fill=g\x,text opacity=1,draw] (z\i\jo) at (\i * \a,\j) {$z(u_\x,\y)$} ;
      }
      \foreach \i in {1,2,3}{
        \node[draw, thick, rectangle, rounded corners, fit=(z\i1) (z\i9)] (vi\i) {} ;
      }
      \node at (\a,7.8) {$V_{i,i(1)}$} ;
      \node at (2*\a,7.8) {$V_{i,i(2)}$} ;
      \node at (3*\a,7.8) {$V_{i,i(3)}$} ;
      \foreach \i/\j in {1/3,2/4,3/6,4/8}{
        \foreach \h in {\j,...,9}{
          \draw (l\i.east) -- (z1\h.west) ;
        }
      }
      \foreach \i/\j in {1/4,2/4,3/6,4/8,5/8,6/9,7/9}{
        \foreach \h in {\j,...,9}{
          \draw (z1\i.east) -- (z2\h.west) ;
        }
      }
      \foreach \i/\j in {1/2,2/2,3/2,4/4,5/4,6/6,7/6,8/8}{
        \foreach \h in {\j,...,9}{
          \draw (z2\i.east) -- (z3\h.west) ;
        }
      }
      \foreach \i/\j in {1/2,2/3,3/3,4/4,5/4,6/5,7/5}{
        \foreach \h in {\j,...,\t}{
          \draw (z3\i.east) -- (r\h.west) ;
        }
      }
      \node[very thick,rounded corners,draw] at (0,3) {$l(u_3)$} ;
      \node[very thick,rounded corners,draw] at (12.6,3) {$r(u_3)$} ;
      \node[very thick,rounded corners,draw] at (\a,2) {$z(u_3,v_3)$} ;
      \node[very thick,rounded corners,draw] at (2*\a,4) {$z(u_3,w_1)$} ;
      \node[very thick,rounded corners,draw] at (3*\a,3) {$z(u_3,x_2)$} ;
    \end{tikzpicture}
    \caption{The encoding $H_i$ of one $V_i$ ordered $u_1 < u_2 < u_3 < u_4 < u_5$. In bold, a possible independent set intersecting the five sets and containing a consistent pair $l(u), r(u)$.}
    \label{fig:encoding-Vi}
  \end{figure}
  
  For each $ij \in E(H)$, we create $|E(V_i,V_j)|$ copies of the binomial tree $T_5$.
  So these trees are in one-to-one correspondence with the edges of $G$ between $V_i$ and $V_j$, and we denote by $T_5(uv)$ the tree corresponding to $uv \in E(V_i,V_j)$.
  We denote by $\beta(uv)$ and $\gamma(uv)$ the two children getting color 2 of the only two vertices colored 3, in the Grundy coloring of $T_5(uv)$ which gives color 5 to its root.
  We remove the pendant neighbor of $\beta(uv)$ and of $\gamma(uv)$ (the two vertices getting color 1 and supporting $\beta(uv)$ and $\gamma(uv)$).
  This results in a fourteen-vertex tree.
  We denote this set of trees by $\mathcal T_{i,j}$, and the $|E(V_i,V_j)|$ roots of the $T_5$ by $\mathcal R_{i,j}$.
  For each $ij \in E(H)$ and for every pair $z(u,v) \in V_{i,j}, \, z(v,u) \in V_{j,i}$, we make $z(u,v)$ and $\beta(uv)$ adjacent, and we make $z(v,u)$ and $\gamma(uv)$ adjacent.

  For every $i \in [k]$, we create $|V_i|$ copies of the binomial tree $T_5$.
  These trees are in one-to-one correspondence with $V_i$.
  Similarly as above, we denote by $\beta(u)$ and $\gamma(u)$ the two vertices getting color 2, whose parents are colored 3, in $T_5(u)$ and we remove their pendant neighbor (colored 1).
  For every pair $l(u) \in L_i$ and $r(u) \in R_i$, we link $l(u)$ and $\beta(u)$, and we link $r(u)$ and $\gamma(u)$.
  We denote this set of trees by $\mathcal T_i$, and the $|V_i|$ roots of the $T_5$ by $\mathcal R_i$.

  We finally create one copy of the binomial tree $T_q$.
  We observe that there are $|E(H)|$ sets $\mathcal R_{i,j}$ and $|V(H)|$ sets $\mathcal R_i$.
  The binomial tree $T_q$ has at least $|V(H)|+|E(H)|=2.5k$ vertices getting color $7$ in the greedy coloring giving color $q$ to the root.
  Indeed the number of vertices colored 7 is $2^{q-8}$, and it holds that $q - 8 \geqslant \log k + \log{2.5}$.
  We map $2.5k$ distinct vertices colored 6 in $T_q$, that are children of vertices colored 7, in a one-to-one correspondence with $V(H) \cup E(H)$.
  Let $f(i)$ be the vertex mapped to $i \in V(H)$ and $f(ij)$ be the vertex mapped to $ij \in E(H)$.
  We further remove the subtree $T_5$ of each of these $2.5k$ vertices colored 6. 
  For every $i \in V(H)$, we link $f(i)$ to all the vertices in $\mathcal R_i$.
  Similarly for every $ij \in E(H)$, we link $f(ij)$ to all the vertices in $\mathcal R_{i,j}$.
  This finishes the construction of the graph $G'$.
  Solving \gr in time $f(q)n^{o(2^{q - \log q})} = f(\lceil \log k \rceil+55)n^{o(k / \log k)}$ would give the same running time for \kmSI, which is ruled out under the ETH. 
 We now prove that the reduction is correct. 
  \paragraph*{A solution to \kmSI implies $\mathbf{\Gamma(G') \geqslant q}$.}
  Let $v_1 \in V_1, v_2 \in V_2, \ldots, v_k \in V_k$ be a fixed solution to the \kmSI-instance (the colored isomorphism being $i \in [k] \mapsto v_i$).
  We say that each edge $v_iv_j$ is \emph{in the solution} (for $i \neq j \in [k]$).
  We color 1 all the vertices of $G'$ corresponding to edges in the solution, that is, all the vertices $z(v_i,v_j)$, as well as all vertices of $G'$ corresponding to vertices in the solution, that is $l(v_i)$ and $r(v_i)$.
  This is possible since the five vertices $l(v_i), z(v_i,v_{i(1)})$, $z(v_i,v_{i(2)})$, $z(v_i,v_{i(3)}), r(v_i)$ form an independent set since $\neg(v_i <_i v_i)$.

  We can now color 2 the vertices $\beta(v_i)$ and $\gamma(v_i)$.
  Therefore the root of $T_5(v_i)$ can receive color 5.
  Moreover, for every $ij \in E(H)$ we can color 2 the vertices $\beta(v_iv_j)$ and $\gamma(v_iv_j)$.
  Therefore the root of $T_5(v_iv_j)$ can receive color 5.
  Since one vertex in each $\mathcal R_i$, and one vertex in each $\mathcal R_{i,j}$ get color 5, the vertices $f(i)$ and $f(ij)$ can all get color 6.
  Finally the root of $T_q$ can receive color $q$.    
  \paragraph*{$\mathbf{\Gamma(G') \geqslant q}$ implies a solution to \kmSI.}
  We first show that only the two vertices of $T_q$ with degree $q-1$ can get color $q$.
  Besides these two vertices, the only vertices of $T_q$ with sufficiently large degree to get color $q$ are the vertices $f(i)$ and $f(ij)$.
  But these vertices have at most one neighbor of degree more than $5$.
  So according to \cref{cor:fewHighDegNeighbors}, they cannot receive a color higher than $7 < q$.
  Now we use \cref{lem:chain-of-length-four} to bound the color reachable outside of $T_q$.
  For every $i \in [k]$, the induced subgraph $G'[H_i]$ is a length-four path of half-graphs.
  Thus by \cref{lem:twin,lem:chain-of-length-four}, $\Gamma(G'[H_i]) \leqslant 53$.
  All the vertices in the open neighborhood of $V_{i,i(1)} \cup V_{i,i(2)} \cup V_{i,i(3)}$ have degree at most 2.
  So by \cref{lem:crownSmallDeg} vertices outside $T_q$ cannot receive a color beyond $55 < q$.

  We now established that if $\Gamma(G') \geqslant q$ (actually $\Gamma(G') = q$), then either one of the two possible roots of $T_q$ shall receive color $q$.
  By \cref{lem:subtree}, this implies that all the vertices $f(i)$ and $f(ij)$ receive color 6, and that in each $\mathcal R_i$ and each $\mathcal R_{i,j}$ there is at least one vertex receiving color 5.
  For every $i \in [k]$, let $T_5(u_i)$ be one $T_5$ of $\mathcal T_i$ whose root gets color 5.
  We will now show that $\{u_1, \ldots, u_i, \ldots, u_k\}$ is a solution to the \kmSI-instance.
  Again by \cref{lem:subtree}, this is only possible if $\beta(u_i)$ and $\gamma(u_i)$ both get color 2, and their unique neighbor outside $T_5(u_i)$ gets color 1.
  It means that $l(u_i)$ and $r(u_i)$ both get color 1.

  Since every $\mathcal R_{i,j}$ contains at least one vertex colored 5, \cref{lem:subtree} implies that every $V_{i,i(p)}$ (for each $p \in [3]$) gets at least one vertex colored 1.
  Let $z(u,v) \in V_{i,i(1)}$, $z(u',v') \in V_{i,i(2)}$, and $z(u'',v'') \in V_{i,i(3)}$ three vertices getting color 1.
  As $\{l(u_i), z(u,v), z(u',v'), z(u'',v''), r(u_i)\}$ should be an independent set, we have $u_i \geqslant_i u \geqslant_i u' \geqslant_i u'' \geqslant_i u_i$.
  This implies that $u_i = u = u' = u''$.
  In turn that implies that no vertex $z(u^*,v) \in V_{i,i(1)} \cup V_{i,i(2)} \cup V_{i,i(3)}$ with $u^* \neq u_i$ can get color 1.
  Indeed $l(u_i)$ prevents a 1 ``above'' $z(u_i,v) \in V_{i,i(1)}$ and $z(u_i,v') \in V_{i,i(2)}$ prevents a 1 ``below'' $z(u_i,v)$.
  The same goes for the color classes $V_{i,i(2)}$ and $V_{i,i(3)}$.
  Thus the only trees $T_5(uv) \in \mathcal T_{i,j}$ that can get color 5 at their root are the ones such that $\{u,v\} \subset \{u_1, \ldots, u_k\}$.
  As all the $1.5k$ sets $\mathcal T_{i,j}$ have such a tree, it implies that $\{u_1, \ldots, u_k\}$ is a solution to the \kmSI-instance. 
\end{proof}

\section{Parameterized hardness of \bchr}\label{sec:core}

A length-two path of half-graphs have arbitrary large $b$-chromatic core.
Indeed the coloring of \cref{subfig:cycle-HG} yields a $b$-chromatic core achieving the height of the half-graphs number of colors (even without the edges of the ``cycle'').
Nevertheless a simple half-graph only admits $b$-chromatic cores of bounded size.
We show how to still build a W[1]-hardness construction in this furtherly constrained situation. 

\begin{theorem}\label{hardness:bchr}
\bchr is W[1]-complete.
\end{theorem}

\begin{proof}
  The inclusion in W[1] is immediate by the characterization of Cesati~\cite{Cesati03}, and the facts that minimal witnesses have size at most $k^2$, and that given the subgraph induced by a minimal witness one can check if it is  solution. 
  To show W[1]-hardness, we reduce from \textsc{$k$-by-$k$ Grid Tiling}. 
  We recall that, in this problem, given $k^2$ sets of pairs over $[n]$, say, $(P_{i,j} \subseteq [n] \times [n])_{i,j \in [k] \times [k]}$, "displayed in a $k$-by-$k$ grid", one has to find one pair $(x_{i,j},y_{i,j})$ in each $P_{i,j}$ such that $x_{i,j} = x_{i,j+1}$ and $y_{i,j} = y_{i+1,j}$, for every $i, j \in [k-1]$.
  This problem remains W[1]-hard under these assumptions.

  \paragraph*{Construction.}
  Let $(P_{i,j} \subseteq [n] \times [n])_{i,j \in [k] \times [k]}$ be the instance of \textsc{Grid Tiling}.
 For each $(i,j)$, we have the set of pairs $P_{i,j}$ with $\lvert P_{i,j} \rvert = t$.
  For each $(i,j)$, we add a biclique $K_{t,q-9}(i,j) := K_{t,q-9}$, where $q := 14k^2$.
  The  part of $K_{t,q-9}(i,j)$ with size $t$ is denoted by $A_{i,j}$ and the other part by $B_{i,j}$ (see \cref{fig:bicliqueAndHG}). 
  We denote by $A_{i,j}$ the $t$ vertices to the left of $K_{t,q-9}(i,j)$ on \cref{fig:bicliqueAndHG}, and by $B_{i,j}$, the $q-9$ vertices to the right.
  The vertices of $A_{i,j}$ are in one-to-one correspondence with the pairs of $P_{i,j}$.
  We denote by $a_{i,j}(x,y) \in A_{i,j}$ the vertex corresponding to $(x,y) \in P_{i,j}$.
  We make the construction ``cyclic'', or rather ``toroidal''.
  So in what follows, every occurrence of $i+1$ or $j+1$ should be interpreted as 1 in case $i=k$ or $j=k$.
  
  For every vertically (resp.~horizontally) consecutive pairs $(i,j)$ and $(i+1,j)$ (resp.~$(i,j)$ and $(i,j+1)$) we add a half-graph $H(i \rightarrow i+1,j)$ (resp.~$H(i,j \rightarrow j+1)$) with bipartition $H(\underline{i} \rightarrow i+1,j) \cup H(i \rightarrow \underline{i+1},j)$ (resp.~$H(i,\underline{j} \rightarrow j+1) \cup H(i,j \rightarrow \underline{j+1})$).
  Both sets $H(\underline{i} \rightarrow i+1,j)$ and $H(i,\underline{j} \rightarrow j+1)$ are in one-to-one correspondence with the vertices of $A_{i,j}$, while the set $H(i \rightarrow \underline{i+1},j)$ is in one-to-one correspondence with the vertices of $A_{i+1,j}$, and $H(i,j \rightarrow \underline{j+1})$, with the vertices of $A_{i,j+1}$.
  We denote by $h_{\underline{i} \rightarrow i+1,j}(x,y)$ (resp.~$h_{i \rightarrow \underline{i+1},j}(x',y')$) the vertex corresponding to $a_{i,j}(x,y)$ (resp.~$a_{i+1,j}(x',y')$).
  Similarly we denote by $h_{i,\underline{j} \rightarrow j+1}(x,y)$ (resp.~$h_{i,j \rightarrow \underline{j+1}}(x',y')$) the vertex corresponding to $a_{i,j}(x,y)$ (resp.~$a_{i,j+1}(x',y')$).
  Every vertex in a half-graph $H(i \rightarrow i+1,j)$ or $H(i,j \rightarrow j+1)$ is made adjacent to its corresponding vertex in $A_{i,j} \cup A_{i+1,j} \cup A_{i,j+1}$.
  Thus $a_{i,j}(x,y)$ is linked to $h_{\underline{i} \rightarrow i+1,j}(x,y)$, $h_{i-1 \rightarrow \underline{i},j}(x,y)$, $h_{i,\underline{j} \rightarrow j+1}(x,y)$, and $h_{i,j-1 \rightarrow \underline{j}}(x,y)$.
  Note that underlined numbers are used to distinguish names, and to give information on its neighborhood.
  We call \emph{vertical} half-graph an $H(i \rightarrow i+1,j)$, and \emph{horizontal} half-graph an $H(i,j \rightarrow j+1)$.
  We now precise the order of the half-graphs.
  In vertical half-graphs, we put an edge between $h_{\underline{i} \rightarrow i+1,j}(x,y)$ and $h_{i \rightarrow \underline{i+1},j}(x',y')$ whenever $y < y'$.
  In horizontal half-graphs, we put an edge between $h_{i, \underline{j} \rightarrow j+1}(x,y)$ and $h_{i,j \rightarrow \underline{j+1}}(x',y')$ whenever $x < x'$.

  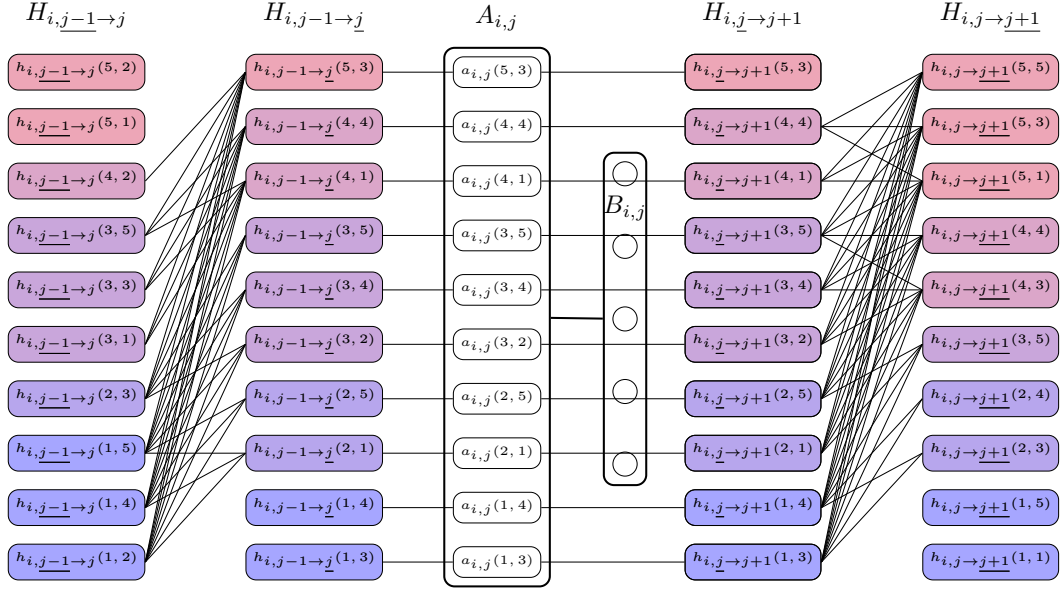
\begin{figure}[h!]
    \centering
    \resizebox{400pt}{!}{
    \begin{tikzpicture}
      \def\t{10}
      \def\q{5}
      \def\n{5}
      \def\vs{0.75}
      \foreach \i/\x/\y in {1/1/3,2/1/4,3/2/1,4/2/5,5/3/2,6/3/4,7/3/5,8/4/1,9/4/4,10/5/3}{
        \node[draw,rectangle,rounded corners,inner sep=0.1cm] (a\i) at (0.75, \i * \vs) {\tiny{$a_{i,j}(\x,\y)$}} ;
        \node[fill opacity=0.35,fill=g\x,rectangle,rounded corners,inner sep=0.1cm] at (-1.75, \i * \vs) {\tiny{$h_{i,j-1 \rightarrow \underline{j}}(\x,\y)$}} ;
        \node[draw,rectangle,rounded corners,inner sep=0.1cm] (hp1\i) at (-1.75, \i * \vs) {\tiny{$h_{i,j-1 \rightarrow \underline{j}}(\x,\y)$}} ;
        \node[fill opacity=0.35,fill=g\x,draw,rectangle,rounded corners,inner sep=0.1cm] at (4.25, \i * \vs) {\tiny{$h_{i,\underline{j} \rightarrow j+1}(\x,\y)$}} ;
        \node[draw,rectangle,rounded corners,inner sep=0.1cm] (h2\i) at (4.25, \i * \vs) {\tiny{$h_{i,\underline{j} \rightarrow j+1}(\x,\y)$}} ;
        \draw (hp1\i.east) -- (a\i.west) ;
        \draw (a\i.east) -- (h2\i.west) ;
      }
      \node at (0.75, \t * \vs + \vs) {$A_{i,j}$} ;
      \node at (-1.75, \t * \vs + \vs) {$H_{i,j-1 \rightarrow \underline{j}}$} ;
      \node at (4.25, \t * \vs + \vs) {$H_{i,\underline{j} \rightarrow j+1}$} ;
      \foreach \i in {1,...,\q}{
        \node[draw,circle] (b\i) at (2.5,\i + 1.1) {} ;
      }
      \node at (2.5, 7.5 * \vs) {$B_{i,j}$} ;
      \node[draw,rectangle,thick,rounded corners,fit=(a1) (a\t)] (aij) {} ;
      \node[draw,rectangle,thick,rounded corners,fit=(b1) (b\q)] (bij) {} ;
      \draw[thick] (aij) -- (bij) ;
      \foreach \i/\x/\y in {1/1/2,2/1/4,3/1/5,4/2/3,5/3/1,6/3/3,7/3/5,8/4/2,9/5/1,10/5/2}{
        \node[fill opacity=0.35,fill=g\x,rectangle,rounded corners,inner sep=0.1cm] at (-5, \i * \vs) {\tiny{$h_{i,\underline{j-1} \rightarrow j}(\x,\y)$}} ;
        \node[draw,rectangle,rounded corners,inner sep=0.1cm] (h1\i) at (-5, \i * \vs) {\tiny{$h_{i,\underline{j-1} \rightarrow j}(\x,\y)$}} ;
      }
      \foreach \i/\x/\y in {1/1/1,2/1/5,3/2/3,4/2/4,5/3/5,6/4/3,7/4/4,8/5/1,9/5/3,10/5/5}{
        \node[fill opacity=0.35,fill=g\x,rectangle,rounded corners,inner sep=0.1cm] at (7.5, \i * \vs) {\tiny{$h_{i, j \rightarrow \underline{j+1}}(\x,\y)$}} ;
        \node[draw,rectangle,rounded corners,inner sep=0.1cm] (hp2\i) at (7.5, \i * \vs) {\tiny{$h_{i, j \rightarrow \underline{j+1}}(\x,\y)$}} ;
      }
      \node at (-5, \t * \vs + \vs) {$H_{i,\underline{j-1} \rightarrow j}$} ;
      \node at (7.5, \t * \vs + \vs) {$H_{i,j \rightarrow \underline{j+1}}$} ;
      \foreach \i/\j in {1/3,2/3,3/3,4/5,5/8,6/8,7/8,8/10}{
        \foreach \h in {\j,...,\t}{
          \draw (h1\i.east) -- (hp1\h.west) ;
        }
      }
      \foreach \i/\j in {1/3,2/3,3/5,4/5,5/6,6/6,7/6,8/8,9/8}{
        \foreach \h in {\j,...,\t}{
          \draw (h2\i.east) -- (hp2\h.west) ;
        }
      }
    \end{tikzpicture}
    }
    \caption{The biclique $K_{t,q-9}(i,j)$ encoding the pairs $P_{i,j}$, and its connection to the two neighboring horizontal half-graphs, with $n=5$, $t=10$, and $q=14$.}
    \label{fig:bicliqueAndHG}
  \end{figure}
  
  We then add a global clique $C$ of size $q-k^2$.
  We attach $k^2$ private neighbors to each vertex of $C$.
  Among the $q-k^2$ vertices of $C$, we arbitrarily distinguish 33 vertices: a set $D=\{d_1, \ldots, d_{18}\}$ of size 18, and three sets $C', C^-, C^+$ each of size 5. 
  We fully link $d_z$ to every $B_{i,j}$ if $z$ takes one of the following values: 
\begin{itemize}
  \item $3(j \text{ mod } 3 - 1)+ i \text{ mod } 3,$
 \item $\text{succ}(3(j \text{ mod } 3 - 1)+ i \text{ mod } 3),$
  \item $3(i \text{ mod } 3 - 1) + j \text{ mod } 3 + 9,$
  \item $\text{succ}(3(i \text{ mod } 3 - 1) + j \text{ mod } 3 + 9),$
\end{itemize}   
where the modulos are always taken in $\{0,1,2\}$, and $\text{succ}(x) := x+1$ if $x$ is not dividable by $3$ and $\text{succ}(x) := x-2$ otherwise (see \cref{fig:forbiddenColors}).
Note that each $B_{i,j}$ is linked with $d_z$ for two successive (indicated by the operator $\text{succ}(x)$) integers $z$ in the range of $[1,3]$, $[4,6]$ or $[7,9]$  
depending on the  coordinate $j$ modulo $3$. Likewise, each $B_{i,j}$ is linked with $d_z$ for two successive integers $z$ in the range of $[10,12]$, $[13,15]$ or $[17,18]$  
depending on the  coordinate $i$ modulo $3$.

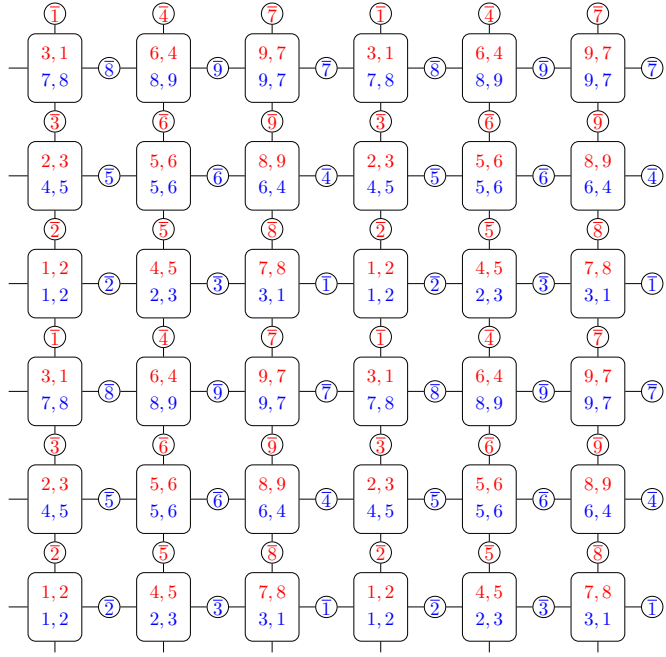
\begin{figure}[h!]
    \centering
    \resizebox{250pt}{!}{
    \begin{tikzpicture}
      \def\s{2}
      \foreach \g/\h in {0/1,1/2,2/3,0/4,1/5,2/6}{
        \foreach \p in {0,1}{
          \pgfmathtruncatemacro\gi{1+3*\g}
          \pgfmathtruncatemacro\gj{2+3*\g}
          \pgfmathtruncatemacro\gk{3+3*\g}
          \pgfmathtruncatemacro\pi{1+3*\p}
          \pgfmathtruncatemacro\pj{2+3*\p}
          \pgfmathtruncatemacro\pk{3+3*\p}
          \pgfmathtruncatemacro\xp{(1+3*\p) * \s}
          \node[blue] (c\h\pi) at (\xp,\h * \s) {$\gi, \gj$} ;
          \node[blue] (c\h\pj) at (\xp + \s,\h * \s) {$\gj, \gk$} ;
          \node[blue] (c\h\pk) at (\xp + 2 * \s,\h * \s) {$\gk, \gi$} ;
 
          \pgfmathsetmacro\yp{(1+3*\p) * \s + \s / 4}
          \node[red] (d\h\pi) at (\h * \s, \yp) {$\gi, \gj$} ;
          \node[red] (d\h\pj) at (\h * \s, \yp + \s) {$\gj, \gk$} ;
          \node[red] (d\h\pk) at (\h * \s, \yp + 2 * \s) {$\gk, \gi$} ;

          \node[draw,circle,inner sep=0.02cm] (f\h\pi) at (\xp + \s / 2,\h * \s + \s / 8) {\textcolor{blue}{$\overline{\gj}$}} ;
          \node[draw,circle,inner sep=0.02cm] (f\h\pj) at (\xp + \s + \s / 2,\h * \s + \s / 8) {\textcolor{blue}{$\overline{\gk}$}} ;
          \node[draw,circle,inner sep=0.02cm] (f\h\pk) at (\xp + 2 * \s + \s / 2,\h * \s + \s / 8) {\textcolor{blue}{$\overline{\gi}$}} ;

          \node[draw,circle,inner sep=0.02cm] (g\h\pi) at (\h * \s, \yp + 0.38 * \s) {\textcolor{red}{$\overline{\gj}$}} ;
          \node[draw,circle,inner sep=0.02cm] (g\h\pj) at (\h * \s, \yp + \s + 0.38 *\s) {\textcolor{red}{$\overline{\gk}$}} ;
          \node[draw,circle,inner sep=0.02cm] (g\h\pk) at (\h * \s, \yp + 2 * \s + 0.38 *\s) {\textcolor{red}{$\overline{\gi}$}} ;
        }
      }
      \foreach \h in {1,...,6}{
        \foreach \y in {1,...,6}{
          \node[draw,rectangle,rounded corners,fit=(c\h\y) (d\y\h)] (e\h\y) {} ;
        }
      }
      \foreach \h in {1,...,6}{
        \foreach \y in {1,...,6}{
          \draw (e\h\y) -- (f\h\y) ;
          \draw (e\h\y) -- (g\y\h) ;
        }
      }
      \foreach \h in {1,...,5}{
        \foreach \y in {1,...,6}{
          \pgfmathtruncatemacro{\hp}{\h+1}
          \draw (e\hp\y) -- (g\y\h) ;
        }
      }
      \foreach \h in {1,...,6}{
        \draw (e\h1.west) --++(-0.35,0) ;
        \draw (e1\h.south) --++(0,-0.2) ;
        \foreach \y in {1,...,5}{
          \pgfmathtruncatemacro{\yp}{\y+1}
          \draw (e\h\yp) -- (f\h\y) ;
        }
      }
    \end{tikzpicture}
    }
    \caption{The rounded rectangles represent the $B_{i,j}$, and the numbers therein, the $z \in [18]$ such that $d_z$ is fully linked to it. These are the ``colors'' that a center in $A_{i,j}$ will have to fight for. The circles represent the half-graphs, and the number therein, the only $z$ such that $d_z$ is \emph{not} linked to it. Red integers are offset by 9 ($\textcolor{red}{1}=\textcolor{blue}{10}$, $\textcolor{red}{2}=\textcolor{blue}{11}$, and so on). The edges represent the non-empty interaction between the $A_{i,j}$ and the half-graphs. The structure is glued like a torus.}
      \label{fig:forbiddenColors}
  \end{figure}

  We observe that $B_{i,j}$ and $B_{i+1,j}$ are linked to exactly one common $d_z$ ($z \in [18]$), and we fully link the half-graph $H(i \rightarrow i+1,j)$ to $D \setminus \{z\}$.
  Similarly we fully link the half-graph $H(i,j \rightarrow j+1)$ to $D \setminus \{z\}$ where $z$ is the unique integer of $[18]$ such that $d_z$ is fully linked to both $B_{i,j}$ and $B_{i,j+1}$.
  We fully link each $A_{i,j}$ to $C'$, each $H(\underline{i} \rightarrow i+1,j)$ and $H(i,\underline{j} \rightarrow j+1)$ to $C^-$, and each $H(i \rightarrow \underline{i+1},j)$ and $H(i,j \rightarrow \underline{j+1})$ to $C^+$.
  This is just to prevent that one uses a vertex of a $B_{i,j}$ or of a half-graph as a center.
  As we will see, intended solutions have all their centers in $C \cup \bigcup_{i,j \in [k]} A_{i,j}$.
  
  This ends our polytime construction.
  We denote by $G$ the obtained graph.
  We ask for a $b$-chromatic core achieving color $q$.
  This polytime construction is of polynomial size in the \textsc{$k$-by-$k$ Grid Tiling}-instance, and $q$ is a computable function of $k$.
  We thus only need to show that the reduction is correct.

    \paragraph*{A solution to \textsc{$k$-by-$k$ Grid Tiling} implies $\mathbf{\Gamma_b(G) \geqslant q}$.}
  We color the clique $C$ with all the colors in $[q-k^2]$, and the $k^2$ private neighbors of each vertex of $C$ with the corresponding remaining colors in $[q-k^2+1,q]$.
  The coloring of $D$ uses all the colors in $[18]$ and respects \cref{fig:forbiddenColors}.
  The colors in $[q-k^2]$ are already \emph{realized}.
  The remaining $k^2$ colors in $[q-k^2+1,q]$ will be realized by the $k^2$ centers corresponding to the solution of the \textsc{Grid Tiling}-instance.
  Let $\{p_{i,j}=(x_{i,j},y_{i,j}) \in P_{i,j}\}_{i,j \in [k]}$ be a grid-tiling solution. 
  We arbitrarily color the $a_{i,j}(x_{i,j},y_{i,j})$ with all the colors of $[q-k^2+1,q]$. 
  The intended $q$-witness shall take the vertices $a_{i,j}(x_{i,j},y_{i,j})$ as the centers of the top $k^2$ colors. 
  Note that each $a_{i,j}(x_{i,j},y_{i,j})$ has precisely $q$ neighbors: four vertices 
  corresponding to $a_{i,j}(x_{i,j},y_{i,j})$ in each of the four half-graphs, $q-9$ vertices in $B_{i,j}$ and five vertices in $C'$. 
  Therefore, in order for $a_{i,j}(x_{i,j},y_{i,j})$ to become a center, we color the set $B_{i,j}$ with 
  the $q-9$ colors out of $[q]$ which are the only available colors after excluding: 
  the five colors of $C'$ that are already seen by $A_{i,j}$, and the four colors of $[18]$ already seen by $B_{i,j}$.  
  
  The four colors left to be seen by $a_{i,j}(x_{i,j},y_{i,j})$ shall be provided by the four half-graphs.
  Let $z \in [18]$ be the common missing color of $a_{i,j}(x_{i,j},y_{i,j})$ and $a_{i+1,j}(x_{i+1,j},y_{i+1,j})$ (resp.~ $a_{i,j+1}(x_{i,j},y_{i,j+1})$): this is the color adjacent to both $B_{i,j}$ and $B_{i+1,j}$ (resp.~$B_{i,j}$ and $B_{i,j+1}$).
  We color by $z$ the neighbor of $a_{i,j}(x_{i,j},y_{i,j})$ in $H(\underline{i} \rightarrow i+1,j)$ and the neighbor of $a_{i+1,j}(x_{i+1,j},y_{i+1,j})$ in $H(i \rightarrow \underline{i+1},j)$ (resp.~ $a_{i,j}(x_{i,j},y_{i,j})$ in $H(i,\underline{j} \rightarrow j+1)$ and the neighbor of $a_{i,j+1}(x_{i,j+1},y_{i,j+1})$ in $H(i,j \rightarrow \underline{j+1})$).
  This is a proper coloring since the $k^2$ pairs form a solution for \textsc{Grid Tiling}, and by construction of the half-graphs.
  Indeed in horizontal (resp.~vertical) half-graphs $H_{i,j \rightarrow j+1}$ (resp.~$H_{i \rightarrow i+1,j}$) the two neighbors are $h_{i,\underline{j} \rightarrow j+1}(x_{i,j},y_{i,j})$ and $h_{i,j \rightarrow \underline{j+1}}(x_{i,j+1},y_{i,j+1})$ (resp.~$h_{\underline{i} \rightarrow i+1,j}(x_{i,j},y_{i,j})$ and $h_{i \rightarrow \underline{i+1},j}(x_{i+1,j},y_{i+1,j})$) which are non-adjacent since $x_{i,j}=x_{i,j+1}$ (resp.~$y_{i,j}=y_{i+1,j}$).
  
  \paragraph*{$\mathbf{\Gamma_b(G) \geqslant q}$ implies a solution to \textsc{$k$-by-$k$ Grid Tiling}.}
  We set $C_d := D \cup C' \cup C^- \cup C^+ \subset C$.
  Let us first show that a vertex of $C \setminus C_d$ has to be a center.
  We thus upperbound the maximum number of centers outside $C \setminus C_d$. 
  We need the following lemma which is true in any graph, not only in the constructed graph $G$.
  Informally it says that vertices with few private neighbors cannot provide too many centers.
  
  \begin{lemma}\label{lem:almost-twins}
    Let $p$ be a non-negative integer, $A$ be a subset of vertices, and $Y$ be their common neighborhood, that is $\bigcap_{v \in A}N(v)$.
    If for every $v \in A$, $\lvert N(v) \setminus Y \rvert \leqslant p$, then the maximum number of centers in $A$ with distinct colors is at most $p+1$.
  \end{lemma}
  \begin{proof}
    For the sake of contradiction, assume that there are $p+2$ centers $v_1, v_2, \ldots, v_{p+2}$ with distinct colors in $A$, say, without loss of generality, $1, 2, \ldots, p+2$.
    Recall that in the context of \bchr (unlike Grundy colorings) all the colors play the same role.
    We reach a contradiction by considering any $v_i$, say $v_1$.
    Vertex $v_1$ colored 1, has to be adjacent to some vertices colored $2, 3, \ldots, p+2$.
    None of these vertices can be in $Y$ since otherwise, the coloring would not be proper.
    But $v_1$ has at most $p$ other neighbors, to accommodate supporting vertices with $p+1$ colors.
  \end{proof}

  Observe that the vertices in an $A_{i,j}$ have only four neighbors (one per neighboring half-graph) that are not in the common neighborhood of $A_{i,j}$.
  Hence \cref{lem:almost-twins} implies that each $A_{i,j}$ has at most $5$ centers (with distinct colors).
  The vertices in a $B_{i,j}$ have no private neighbor, so by the same lemma, they can contain at most one center.
  However if the corresponding $A_{i,j}$ has a center, then no vertex of $B_{i,j}$ can be a center.
  Indeed for a vertex in $A_{i,j}$ to be a supported center, one needs to color at least two vertices of $B_{i,j}$ with distinct colors, say $c$ and $c'$.
  This prevents any vertex of $B_{i,j}$ to be supported, since they would miss either $c$ or $c'$.
  So the maximum number of centers in $A_{i,j} \cup B_{i,j}$ is still 5.
  Thus the total number of centers from all the $A_{i,j} \cup B_{i,j}$ is at most $5k^2$.
  
  We now observe that \cref{lem:almost-twins} also works with a half-graph instead of a biclique between $A$ and $Y$.
  \begin{lemma}\label{lem:almost-twins-HG}
    Assume that $p$ is a non-negative integer and $(A,Y)$ induces a half-graph.
    If for every $v \in A$, $\lvert N(v) \setminus Y \rvert \leqslant p$, then the maximum number of centers in $A$ with distinct colors is at most $p+1$.
  \end{lemma}
  \begin{proof}
    Let $p+2$ centers in $A$: $v_1, v_2, \ldots, v_{p+2}$ with distinct colors, say, $1, 2, \ldots, p+2$, and that $v_{p+2}$ has the highest level in the half-graph among $\{v_1, v_2, \ldots, v_{p+2}\}$.
    We recall that it means that $N(v_{p+2}) \cap Y \subseteq N(v_i) \cap Y$ for every $i \in [p+2]$.
    Then the colors $1, 2, \ldots, p+1$, supporting $v_{p+2}$ cannot come from $Y$.
    A contradiction since $v_{p+2}$ has only at most $p$ other neighbors.
  \end{proof}
  Each set $H_{i,\underline{j} \rightarrow j+1}$, $H_{i,j \rightarrow \underline{j+1}}$, $H_{\underline{i} \rightarrow i+1,j}$, $H_{i \rightarrow \underline{i+1},j}$ has, by \cref{lem:almost-twins-HG}, at most 2 supported centers with distinct colors.
  Indeed, outside the other side of their half-graphs and the shared neighborhood $C^-$ or $C^+$, they have only one private neighbor.
  So the total number of supported centers in all the half-graphs is at most $8k^2$.

  Obviously the private neighbors of the vertices of $C$ cannot be centers, since they have degree 1.
  Overall the number of supported centers with distinct colors is $5k^2+8k^2+[C_d|=13k^2+33 < 14k^2=q$.
  This implies that a vertex of $C \setminus C_d$ is indeed a center.
  Such a vertex has degree $q-1$, so all its vertices should be colored (with distinct colors).
  In particular, the whole clique $C$ should be colored.
  Without loss of generality, let us assume that the clique is colored with $[q-k^2]$, and that $D$ receives the colors from 1 to 18, consistently with \cref{fig:forbiddenColors}.
  By coloring appropriately the private neighbors of the vertices $C$, all the colors in $[q-k^2]$ are realized by $C$.
  Since $\Gamma_b(G) \geqslant q$, actually now $\Gamma_b(G) = q$,  we know that there are $k^2$ additional centers in $V(G) \setminus C$.

  We now prove that the remaining centers are in the $A_{i,j}$.
  A vertex in a $B_{i,j}$ cannot be adjacent to a color used for $C'$, since $C'$ and $A_{i,j}$ are fully adjacent and $N(B_{i,j}) \subseteq D \cup A_{i,j}$.
  Hence such a vertex cannot be a supported center.
  Similarly vertices in the half-graphs cannot be centers.
  Indeed, on one side of the half-graph, they would not find all five colors used for $C^+$ which are excluded in the other side of the half-graph.
  And vice versa, on the other side, they would miss some colors of $C^-$ which are unavailable for the first side.
  Therefore all the $k^2$ remaining centers are in the $A_{i,j}$.

  Let $v \in A_{i,j}$ be a supported center colored $c \in [q-k^2+1,q]$.
  We claim that there cannot be another center with a distinct color in the same $A_{i,j}$.
  The only way for $v$ to be a supported center is that $B_{i,j}$ is colored with the $q-9$ colors which are used for $C \setminus (C' \cup D_{i,j})$ where $D_{i,j} \subset D$ is the set of four vertices fully linked to $B_{i,j}$.  
  (Then the four neighbors of $v$ in the half-graphs have to receive all the colors of $D_{i,j}$.)
  In particular, the only colors that another vertex $v' \neq v \in A_{i,j}$ could get are already realized colors.
  This shows that each $A_{i,j}$ contains exactly one supported center.
  
  We denote by $a_{i,j}(x_{i,j},y_{i,j})$ the unique center in $A_{i,j}$.
  We already observed that the supporting vertices with colors in $D_{i,j}$ have to come from the half-graphs.
  Looking at \cref{fig:forbiddenColors}, for a fixed color of $D_{i,j}$ only one of the four neighboring half-graphs can provide such a color.
  This is due to how we linked the $d_z$ with $B_{i,j}$ and the half-graphs.
  Furthermore the common missing color between two consecutive $A_{i,j}$ has to come from the shared half-graph.
  Again since there is a solution to the \bchr-instance, the neighbors of the centers in the half-graphs form an independent set.
  This means that $x_{i,j} \leqslant x_{i,j+1}$ and $y_{i,j} \leqslant y_{i+1,j}$ for every $i, j \in [k]$.
  Due to the cyclicity of the construction, it implies that $x_{i,j} = x_{i,j+1}$ and $y_{i,j} = y_{i+1,j}$ for every $i, j \in [k]$.
  That in turn implies a solution for the \textsc{Grid Tiling}-instance.
\end{proof}

\section{\pgr and \bchr\ on $K_{t,t}$-free graphs}\label{sec:tractability}

In the following subsection, we prove that both \bchr\ and \pgr\ can be solved in FPT time when all but a bounded number of vertices have bounded degree.
This is a preparatory step to show the tractability in $K_{t,t}$-free graphs.

\subsection{FPT algorithm on almost bounded-degree graphs}\label{subsec:degree}

The technique of random separation~\cite{Cai06,Cai08}, inspired by the color coding technique~\cite{Alon1995}, comes handy when one wants to separate a latent vertex-subset 
of small size from the rest of the graph. 
Consider two disjoint vertex-sets $A$ and $B$. 
By sampling a vertex-subset $S$, with probability $1/2$ for each vertex to be included in $S$, the probability that $S$ contains $A$ and avoids $B$ is at least $2^{-\lvert A\rvert - \lvert B\rvert}$. 
Especially if we choose $B$ as the open neighborhood of $A$, this means that with a reasonably large probability, a random sample $S$ will precisely carve out the connected components of $G[A]$ as long as both $A$ and $B$ have bounded size.   
Notice that some of the connected component of $G[S]$, upon a successful sample $S$, might not be a connected component of $G[A]$ and some extra works would be needed to extract the desired set $A$. 
Nonetheless, the knowledge that the components components of $G[A]$ are captured as connected components of $G[S]$ can be substantially useful. 
A derandomized version of random separation can be obtained with the so-called splitters of Naor~et al.~\cite{NaorSS95} (see also Chitnis et al.~\cite{ChitnisCHPP16}). 
For two disjoint sets $A$ and $B$ of a universe $U$, we say that $S\subseteq U$ is an \emph{$(A,B)$-separating set} if $A\subseteq S$ and $B\cap S=\emptyset$. 

\begin{lemma}[Chitnis et al.~\cite{ChitnisCHPP16}]\label{lem:magicsplitter}
Let $a$ and $b$ be non-negative integers. For an $n$-element universe $U$, there exists 
a family $\cF$ of $2^{O(\min (a,b) \log{(a+b)})} \log n$ subsets of $U$  
such that for any disjoint subsets $A,B\subseteq U$ with $\lvert A\rvert\leqslant a$ and $\lvert B\rvert\leqslant b$, there 
exists an $(A,B)$-separating set $S$ in $\cF$. Furthermore, such a family $\cF$ can be constructed in time $2^{O(\min (a,b) \log{(a+b)})} n \log n$.
\end{lemma}

\begin{theorem}\label{thm:almostbounded}
Let $G$ be a graph in which at most $s$ vertices have degree larger than $d$. 
Then whether $G$ has a $k$-witness for \bchr\ (\pgr, respectively) can be decided in FPT time parameterized by $k+d+s$.
\end{theorem}

\begin{proof}
Let $X$ be the set of $s$ vertices of degree larger than $d$. 
In order to explain the algorithm and prove its correctness, it is convenient to assume that $G$ does contain a $k$-witness $H$ for \bchr (or \pgr)  as an induced subgraph. 
We define
$
I := V(H)\cap X,  
A := V(H) \setminus X,$ and 
$B := N(A)\setminus X.
$

We can guess $I$ by considering at most $2^s$ subsets of $X$. 
To find $A$, we use~\cref{lem:magicsplitter}.
From the fact that every vertex of $V\setminus X$ has degree at most $d$ and that $H$ is a minimal $k$-witness, we have $\lvert A \rvert  \leqslant k^2$ and $\lvert B \rvert  \leqslant dk^2$. 
Hence, by Lemma~\ref{lem:magicsplitter} with universe $V(G) \setminus X$, we can compute in time $2^{O(k^2 \log{(k^2 + dk^2)})} n\log n$ a family $\mathcal F$ with $2^{O(k^2\log{(k^2 + dk^2)})} \log n$ subsets of $V(G) \setminus X$, that contains an $(A, B)$-separating set. 

We guess this $(A, B)$-separating set by iterating over all elements of $\mathcal F$.
Let $S$ be a correct guess, i.e., $S$ is an $(A,B)$-separating set. 
So $A \subseteq S$ and $S \cap B = \emptyset$. 
Observe that every connected component of $G[A]$ appears in $G[S]$ as a connected component. 

Let $\cC_S$ be the set of connected components of $G[S]$ of size at most $k^2$. Since $|A| \leqslant k^2$, larger connected component of $G[S]$ are clearly disjoint from $A$.
Moreover, by definition of $B$ and since $S$ is disjoint from $B$, each connected component of $G[A]$ is an element of $\cC_S$.

Since each element of $\cC_S$ has at most $k^2$ vertices, the number of equivalence classes of $\cC_S$ under graph isomorphism is bounded by a function of $k$. 
In fact, the number of equivalence classes under a stronger form of isomorphism is bounded by a function of $k$. 
We define a labeling function $\ell: S \rightarrow 2^{I}$ as $\ell(v) := N(v)\cap I$. 
Let $\sim_S$ be a relation on $\cC_S$ such that, for every $C,C' \in \cC'_S$:
$C\sim_S C'$ if and only if there is a graph isomorphism $\phi:C\rightarrow C'$ with $\ell(v)=\ell(\phi(v))$ for every $v\in C$.
Let $[\sim_S]$ be the  partition of $\cC_S$ into equivalence classes under $\sim_S$. 
As members of $\cC_S$ have cardinality at most $k^2$ and there are $2^{|I|} \leqslant 2^{k^2}$ labels, 
$\cC_S$ has at most $2^{2k^4}$ equivalence classes under $\sim_S.$
And thus we can compute $[\sim_S]$ in time $2^{2k^4}n$.   
The definition of $\sim_S$ clearly implies that two equivalent sets $C$ and $C'$ under $\sim_S$ are exchangeable as a connected component of $H-X$. 
That is, for any induced subgraph $D$ of $G$ with $V(D)\cap X=I$, 
if $C$ is a connected component of $D-I$, then $G[(V(D)\setminus C)\cup C']$ is isomorphic to $D$.

We will now guess, by doing an exhaustive search, how many connected components $H-I$ takes from each part of the partition $[\sim_S]$. 
There are  $2^{2k^4k^2}$ possible such guesses and from the fact that the number of connected components in $H-I$ is at most $k^2$. 
Choose an element (i.e. a connected vertex-set) from each part of $[\sim_S]$ as many times as the current guess suggests (if this is impossible, then discard the current guess) and let $W$ be the union of the chosen connected vertex-sets. 
We  can now verify by brute-force that $G[W\cup I]$ is a $k$-witness for \bchr or \pgr, depending on the problem at hand.

To complete the proof of correctness, note that if we find a $k$-witness for some choice of $I$, $S\in \cF$ and $W$, the input graph $G$ clearly admits a $k$-witness.
One can easily observe that the running time is FPT in $k+d+s$.
\end{proof}

\subsection{FPT algorithm on $K_{t,t}$-free graphs}

In this subsection, we present an FPT algorithm on graphs which do not contain $K_{t,t}$ as a subgraph. 
A key element of this algorithm is a combinatorial result (\cref{prop:witnessKttfree}), which states that if there are many vertices of large degree, then one can always find a $k$-witness. 
Therefore, we may assume that the input graph has a small number of vertices of large degree. 
Using the FPT algorithm on almost bounded-degree graphs given in \cref{subsec:degree}, this implies an FPT algorithm on $K_{t,t}$-free graphs.
Most of this subsection is thus devoted to prove our combinatorial tool.

We begin with a few technical lemmas. 
The first lemma is a simple application of a Ramsey-type argument.

\begin{lemma}\label{lem:antibiclique}
Let $t$ and $N$ be two positive integers with $N \geqslant t$, and let $G$ be a graph on a vertex-set $A \uplus B$ not containing  $K_{t,t}$ as a subgraph. 
If $\lvert A \rvert \geqslant N2^{N+t}$ and $\lvert B \rvert \geqslant N+t$, then there exist two  sets $A' \subseteq A$ and $B'\subseteq B$, 
each of size at least $N$, such that there is no edge between $A'$ and $B'$.
\end{lemma}
\begin{proof}
We may assume that $B=\{b_1,\ldots, b_{N+t}\}$ by discarding some vertices of $B$ if necessary. 
We are going to construct a sequence $A_1\supseteq A_2 \supseteq \cdots \supseteq A_{N+t}$ of subsets  of $A$ such that 
for each $i$, the vertex $b_i$ is either complete or anti-complete with $A_i$, and the last set $A_{N+t}$ contains at least $N$ vertices. 
Take $A_1$ as the set of neighbors of $b_1$ if $\lvert N(b_1)\cap A \rvert \geqslant \lvert A \rvert/2$, and as the set of  non-neighbors of $b_1$ otherwise.
Supposed that a sequence $A_1 \supseteq \cdots \supseteq A_i$ for $i< N+t$ has been constructed. We define $A_{i+1}$ as follows; 
if $|N(b_i)\cap A_i| \geqslant |A_i|/2$, then $A_{i+1}=N(b_i)\cap A_i$, and  $A_{i+1}=A_i\setminus N(b_i)$ otherwise.
As the step $i+1$ halves $A_{i}$ in the worst case, the size of $A$ ensures that  $A_{N+t}\geqslant N$.

Now, observe that any vertex $b_i$ of $B$ which is complete with $A_i$ is also complete with $A_{N+t}$.
Since $G$ is $K_{t,t}$-free and $\lvert A_{N+t} \rvert \geqslant N \geqslant t$, all but at most $t-1$ vertices of $B$ are anti-complete with their associated set $A_i$, and thus with $A_{N+t}$.
Therefore we can find at least $N$ vertices in $B$ that are anti-complete with  $A_{N+t}$.
We set $B'$ as these vertices and $A' := A_{N+t}$.
\end{proof}

The next lemma can be obtained as a corollary of the previous lemma. 
\begin{lemma}\label{lem:extract}
For any integers $k$ and $t$, there exists an integer $M$ such that the following holds:  
given a $K_{t,t}$-free graph $G$ and a partition 
$A_1\uplus \cdots \uplus A_k$ of $V(G)$ such that each  $A_i$ contains at least $M$ vertices, 
there exists either a clique on $k$ vertices, or an independent set of size $k^2$ which contains $k$ vertices from each  $A_i$.
\end{lemma}
\begin{proof}
We assume that $t$ is at least $k$; if not, increase $t$ appropriately. Clearly, $G$ is $K_{t,t}$-free for the new value $t$. 
  By ${k \choose 2}$ successive applications of Lemma~\ref{lem:antibiclique}, 
  we can select for $i=1, \dots, k$ (by choosing $M$ large enough) a set $A'_i \subseteq A_i$ of $2^{3t}$ vertices, such that there is no edge linking $A'_j$ and $A'_{j'}$ for all $j \neq j'$.
Concretely, fix 
\[
M=({^{k} 8})^{t}:=\underbrace{8^{8^{\cdot^{\cdot^{8^t}}}}}_\text{8 repeats $k$ times}
\]
and apply~\cref{lem:antibiclique} between $A_1$ and $A_2$, which returns $A'_1\subseteq A_1$ and $A'_i\subseteq A_i$ 
of respective size at least $({^{k-1)} 8})^{t}$. Now let $A_1\leftarrow A'_1$ and $A_2\leftarrow A'_2$, and execute the same procedure 
between $A_1$ and $A_i$ for $i=3$ up to $i=k$. Note that at the end of $k-1$ applications of~\cref{lem:antibiclique}, the size of $A_1$ is at least $({8})^{t}$ 
and  $A_i$ has size at least $({^{k-1)} 8})^{t}$ for each $i\in [k]\setminus 1$. 
Note that now $A_1$ is anti-complete with every $A_i$ for $i\in [i]\setminus 1$. 
By induction on $k$, we are left with $A_1, \ldots, A_k$ which are pairwise anti-complete and each of which has size at least $8^t$.
Then one can apply Ramsey's theorem on all $A_i$'s to get either a clique on $k$ vertices in some $A_i$, or an independent set of size $2t \geqslant k$ in every $A_i$ for $i\in [k]$. 
\end{proof}

The following statement is proved in \cite{AboulkerBBCHMZ18}.

\begin{lemma}[Aboulker et al.~\cite{AboulkerBBCHMZ18}]\label{lem:smallfrac}
  Let $t$ be a positive integer and let $\epsilon \in (0,1)$.
  Then there is an integer $N(t,\epsilon)$ that satisfies the following: if $H=(V,E)$ is a hypergraph on at least $N(t,\epsilon)$ vertices, where all hyperedges have size at least $\epsilon|V|$, and the intersection of any $t$ hyperedges has size at most $t-1$, then $\lvert E \rvert < t/\epsilon^t$.
\end{lemma}

We are ready to prove the key combinatorial result on $K_{t,t}$-free graphs.

\begin{proposition}\label{prop:witnessKttfree}
Let $t, k$ be positive integers.   
  Let $G$ be a $K_{t,t}$-free graph and let $X \uplus Y$ be a partition of $V(G)$.
  There exist integers $f(t,k)$ and $g(t,k)$ such that the following holds: If $|X| \geqslant f(t,k)$, and $|N_Y(x)| \geqslant g(t,k)$ for every $x \in X$,
  then $G$ contains $kK_{1,k}$ as an induced subgraph.
  In particular, $G$ admits $k$-witnesses for \bchr\ and thus for \pgr. 
\end{proposition}
\begin{proof}
We first observe that $kK_{1,k}$ (even $kK_{1,k-1}$) is a $k$-witness for \bchr: color the $k$ centers with distinct colors, and assign colors from $[k] \setminus \{i\}$ to the leaves of the center colored $i$. 
Let $N(t,1/k)$ and $M$ be the integers defined in~\cref{lem:smallfrac,lem:extract} respectively, and set 
$M' := \max(M,N(t,1/k)), f(t,k) := 2^{2t + k(tk^t + t)}, g(t,k) := 2^{k(tk^t+t)}M'$.

By Ramsey's theorem, any graph on at least $f(t,k)$ vertices admits either a clique of size $2t$ or an independent set of size $k(tk^t + t)$.
Since $G[X]$ (which has at least $f(t,k)$ vertices) is $K_{t,t}$-free, the former outcome is impossible, so it has an independent set of size $k(tk^t + t)$.
It should be noted that the inductive proof of Ramsey's theorem yields a greedy linear-time algorithm which outputs a clique or an independent set of the required size.
Hence we efficiently find an independent set of size $k(tk^t + t)$ in $G[X]$. 
Starting from $j=1$, we now prove the following claim inductively for all $j \leqslant k$.

\begin{quote}
$(\star)$ If $\lvert X \rvert \geqslant j(tk^t + t)$ and 
$|N_Y(x)| \geqslant 2^{j(tk^t+t)}M'$ for every $x \in X$,
then there are $j$ vertices $\{b_1,\ldots , b_j\}\subseteq X$ 
and a family of $j$ disjoint vertex-sets $A_1,\ldots , A_j\subseteq Y$ each of size at least $M'$, such that each $A_i$ are private neighbors 
of $b_i$; that is, the vertices of $A_i$ are adjacent with $b_i$ and not adjacent with any other vertices from $\{b_1, \dots, b_j\}$.
\end{quote}

The claim $(\star)$ trivially holds when $j=1$. Suppose it holds for all integers smaller than $j$, where $2\leqslant j\leqslant k$. 
We may assume that $X$ has precisely $j(tk^t + t)$ vertices by discarding some vertices if its size exceeds the bound. 
For each $\emptyset \neq I \subseteq X$, we define $N_I=\bigcap_{v\in I}N_Y(v)\cap \bigcap_{v\in X\setminus I} Y\setminus N_Y(v)$.
Thus $N_I$ corresponds to all the vertices of $Y$ whose neighborhood in $X$ is exactly $I$.
Observe that the $N_I$'s partition $Y$ and that $N_I$ corresponds to the set of vertices of $Y$ that are complete with $I$ and anti-complete with $X-I$. 

Choose a vertex $x\in X$ that minimizes $|N_Y(x)|$. 
As there are $2^{j(tk^t+t)}$ possible subsets of $X$ and $|N_Y(x)| \geqslant 2^{j(tk^t+t)}M'$, there exists  $I^*\subseteq X$ such that $x \in I^*$ and $|N_{I^*}|\geqslant M'$. 

Let $X_x$ 
be the set of vertices in $X$ adjacent with at least $k$-th fraction of $N_Y(x)$, that is, 
\[X_x=\{v\in X: |N_Y(v)\cap N_Y(x)|\geqslant \frac{|N_Y(x)|}{k}\}.\]
Set $X' = X - (I^* \cup X_x)$,  $Y' = Y - N_Y(x)$ and let $G' = G[X' \cup Y']$. 

We want to apply the induction hypothesis on  $G'$ with respect to $X'$ and $Y'$.  
For this, we need to make sure that it satisfies the conditions of $(\star)$ for $j-1$. 
To prove that $|X'| \geqslant(j-1)(tk^t + t)$, we need to bound the size of $I^*$ and $X_x$. To bound the size of $I^*$, notice that $N_{I^*}$ is complete with $I^*$. Since $|N_{I^*}|\geqslant M\geqslant t$ and $G$ is $K_{t,t}$-free, we conclude that $|I^*|<t$. 
To bound the size of $X_x$, we apply Lemma~\ref{lem:smallfrac} with $\epsilon=1/k$ to the hypergraph on the vertex-set $N_Y(x)$ 
and with hyperedge set $\{N_Y(v)\cap N_Y(x): v\in X\}$. 
Each hyperedge of size at least $\frac{|N_Y(x)|}{k}$ corresponds to a vertex in $X_x$ which gives us the bound $|X_x|\leqslant  tk^t$. Notice that Lemma~\ref{lem:smallfrac} can be legitimately applied on this 
hypergraph as $N_Y(x)$ has at least $M' \geqslant N(t,1/k)$ vertices. 
Therefore, we have 
$|X'| \geqslant |X| -t-tk^t\geqslant (j-1)(tk^t + t). $

It remains to verify that each  $v\in X'$ has at least $2^{(j-1)(tk^t+t)}M'$ neighbors in $Y'$. Indeed
\begin{align*}
|N_{Y'}(v)| &\geqslant |N_Y(v)| - |N_Y(v) \cap N_Y(x)| 
 		 \geqslant |N_Y(v)| - \frac{|N_Y(x)|}{k} 
		 \geqslant |N_Y(v)| - \frac{|N_Y(v)|}{k} \\
			&\geqslant \frac{k-1}{k}2^{j(tk^t+t)}M'
			 \geqslant  2^{(j-1)(tk^t+t)}M'.		
\end{align*}

This proves that $G'$ meets the requirement to apply the induction hypothesis, and thus we can find $\{b_2,\ldots , b_j\}$ and sets $A_2,\ldots , A_j$ in $G'$ as claimed in $(\star)$. 
Observe now that $N_{I^*}$ is anticomplete to $\{b_2,\ldots , b_j\}$ and recall that $|N_{I^*}| \geqslant M'$. Hence, setting $b_1 = x$ and $A_1 = N_{I*}$ complete the proof of $(\star)$. 
Now, applying Lemma~\ref{lem:extract} to the sets $A_1\uplus \cdots \uplus A_k$  given by $(\star)$ gives us either a clique on $k$ vertices or the announced set of stars. 
\end{proof}

Combined with the main result of the previous subsection, this implies that \bchr\ and \pgr\ can be solved in FPT time on $K_{t,t}$-free graphs.
Observe that our algorithm is FPT in the combined parameter $k+t$, which is a stronger than having an FPT algorithm in $k$ when $t$ is a fixed constant.
\begin{theorem}\label{thm:fptKttfree}
There is a function $h$ and an algorithm which, given a graph $G=(V,E)$ not containing $K_{t,t}$ as a subgraph, 
decides whether $G$ admits $k$ \bchr (\pgr, respectively) in time $h(k,t)n^{O(1)}$.
\end{theorem}
\begin{proof}
Let $X\subseteq B$ be the set of all vertices whose degree is at least $g(t,k)+f(t,k)$, where $g(t,k)$ and $f(t,k)$ are the 
integers as in \cref{prop:witnessKttfree}. If $X$ contains at least $f(t,k)$ vertices, then we there exists a $k$-witness 
in $G$ by \cref{prop:witnessKttfree}. If $X$ contains less than $f(t,k)$ vertices, 
the algorithm of Theorem~\ref{thm:almostbounded} can be applied to correctly decide whether $G$ contains a $k$-witness.
\end{proof}

\end{document}